\PassOptionsToPackage{unicode}{hyperref}
\PassOptionsToPackage{hyphens}{url}
\PassOptionsToPackage{dvipsnames,svgnames,x11names}{xcolor}
\documentclass[12pt]{article}
\usepackage{todonotes}
\usepackage{amsmath,amssymb}
\usepackage{enumitem}
\usepackage{iftex}
\ifPDFTeX
  \usepackage[T1]{fontenc}
  \usepackage[utf8]{inputenc}
  \usepackage{textcomp} 
\else 
  \usepackage{unicode-math}
  \defaultfontfeatures{Scale=MatchLowercase}
  \defaultfontfeatures[\rmfamily]{Ligatures=TeX,Scale=1}
\fi
\usepackage{lmodern}
\ifPDFTeX\else  
\fi
\IfFileExists{upquote.sty}{\usepackage{upquote}}{}
\IfFileExists{microtype.sty}{
  \usepackage[]{microtype}
  \UseMicrotypeSet[protrusion]{basicmath} 
}{}
\makeatletter
\@ifundefined{KOMAClassName}{
  \IfFileExists{parskip.sty}{%
    \usepackage{parskip}
  }{
    \setlength{\parindent}{0pt}
    \setlength{\parskip}{6pt plus 2pt minus 1pt}}
}{
  \KOMAoptions{parskip=half}}
\makeatother
\usepackage{xcolor}
\makeatletter
\ifx\paragraph\undefined\else
  \let\oldparagraph\paragraph
  \renewcommand{\paragraph}{
    \@ifstar
      \xxxParagraphStar
      \xxxParagraphNoStar
  }
  \newcommand{\xxxParagraphStar}[1]{\oldparagraph*{#1}\mbox{}}
  \newcommand{\xxxParagraphNoStar}[1]{\oldparagraph{#1}\mbox{}}
\fi
\ifx\subparagraph\undefined\else
  \let\oldsubparagraph\subparagraph
  \renewcommand{\subparagraph}{
    \@ifstar
      \xxxSubParagraphStar
      \xxxSubParagraphNoStar
  }
  \newcommand{\xxxSubParagraphStar}[1]{\oldsubparagraph*{#1}\mbox{}}
  \newcommand{\xxxSubParagraphNoStar}[1]{\oldsubparagraph{#1}\mbox{}}
\fi
\makeatother

\usepackage{longtable,booktabs,array}
\usepackage{calc} 
\usepackage{etoolbox}
\makeatletter
\patchcmd\longtable{\par}{\if@noskipsec\mbox{}\fi\par}{}{}
\makeatother
\IfFileExists{footnotehyper.sty}{\usepackage{footnotehyper}}{\usepackage{footnote}}
\makesavenoteenv{longtable}
\usepackage{graphicx}
\makeatletter
\def\maxwidth{\ifdim\Gin@nat@width>\linewidth\linewidth\else\Gin@nat@width\fi}
\def\maxheight{\ifdim\Gin@nat@height>\textheight\textheight\else\Gin@nat@height\fi}
\makeatother
\setkeys{Gin}{width=\maxwidth,height=\maxheight,keepaspectratio}
\makeatletter
\def\fps@figure{htbp}
\makeatother

\makeatletter
\@ifpackageloaded{caption}{}{\usepackage{caption}}
\AtBeginDocument{%
\ifdefined\contentsname
  \renewcommand*\contentsname{Table of contents}
\else
  \newcommand\contentsname{Table of contents}
\fi
\ifdefined\listfigurename
  \renewcommand*\listfigurename{List of Figures}
\else
  \newcommand\listfigurename{List of Figures}
\fi
\ifdefined\listtablename
  \renewcommand*\listtablename{List of Tables}
\else
  \newcommand\listtablename{List of Tables}
\fi
\ifdefined\figurename
  \renewcommand*\figurename{Figure}
\else
  \newcommand\figurename{Figure}
\fi
\ifdefined\tablename
  \renewcommand*\tablename{Table}
\else
  \newcommand\tablename{Table}
\fi
}
\@ifpackageloaded{float}{}{\usepackage{float}}
\floatstyle{ruled}
\@ifundefined{c@chapter}{\newfloat{codelisting}{h}{lop}}{\newfloat{codelisting}{h}{lop}[chapter]}
\floatname{codelisting}{Listing}

\makeatother
\makeatletter
\@ifpackageloaded{caption}{}{\usepackage{caption}}
\@ifpackageloaded{subcaption}{}{\usepackage{subcaption}}
\makeatother

\ifLuaTeX
  \usepackage{selnolig}  
\fi
\usepackage[]{natbib}
\usepackage{bookmark}

\IfFileExists{xurl.sty}{\usepackage{xurl}}{} 
\hypersetup{
  pdftitle={Title},
  pdfauthor={Author 1; Author 2},
  pdfkeywords={3 to 6 keywords, that do not appear in the title},
  colorlinks=true,
  linkcolor={blue},
  filecolor={Maroon},
  citecolor={Blue},
  urlcolor={Blue},
  pdfcreator={LaTeX via pandoc}}

\newcommand{\anon}{1}

\usepackage{amsfonts,amsmath,amssymb,amsthm}
\usepackage{mathrsfs}
\usepackage{algorithm}
\usepackage{algpseudocode}
\usepackage{threeparttable}
\def\qed{\rule{2mm}{2mm}}
\def\independent{\perp \!\!\! \perp}

\newtheorem{theorem}{Theorem}[section]
\newtheorem{lemma}{Lemma}[section]
\newtheorem{definition}{Definition}[section]

\newtheorem{proposition}{Proposition}[section]
\theoremstyle{definition}

\newtheorem{remark}{Remark}[section]
\newtheorem{assumption}{Assumption}[section]

\AtEndEnvironment{remark}{~\qed}
\AtEndEnvironment{example}{~\qed}

\newcommand{\mycomment}[1]{}
\renewcommand{\baselinestretch}{1.35}

\newcommand\Var{{\rm Var}}

\newcommand\cD{\mathcal{D}}

\newcommand\cP{\mathcal{P}}

\newcommand\bA{\mathbf{A}}

\newcommand\Cov{\mbox{\normalfont Cov}}

\newcommand\cG{\mathcal{G}}
\newcommand\cN{\mathcal{N}}

\newcommand\mE{\mathrm{E}}
\newcommand\mP{\mathrm{P}}
\newcommand\DE{\mathrm{DE}}
\newcommand\IE{\mathrm{IE}}
\begin{document}

\def\spacingset#1{\renewcommand{\baselinestretch}%
{#1}\small\normalsize} \spacingset{1}


\if1\anon
{
  \title{\bf Auto-Doubly Robust Estimation of Causal Effects on a Network}
  \author{Jizhou Liu \thanks{Peking University, HSBC Business School,  \url{jizhou.liu@phbs.pku.edu.cn}}
	\and Dake Zhang \thanks{Shanghai Jiao Tong University, Antai College of Economics and Management, \url{dk.zhang@sjtu.edu.cn}}
	\and Eric J. Tchetgen Tchetgen \thanks{University of Pennsylvania, The Wharton School, \url{ett@wharton.upenn.edu}}}
  \maketitle
} \fi

\if0\anon
{
  \bigskip
  \bigskip
  \bigskip
  \begin{center}
    {\LARGE\bf Title}
\end{center}
  \medskip
} \fi

\bigskip
\begin{abstract}
This paper develops new methods for causal inference in observational studies on a single large network of interconnected units, addressing two key challenges: long-range dependence among units and the presence of general interference. We introduce a novel network version of Augmented Inverse Propensity Weighted, which combines propensity score and outcome models defined on the network to achieve doubly robust identification and estimation of both direct and spillover causal effects. Under a network version of conditional ignorability, the proposed approach identifies the expected potential outcome for a unit given the treatment assignment vector for its network neighborhood up to a user-specified distance, while marginalizing over treatment assignments for the rest of the network.  Under the union of two  Markov assumptions—one governing the propensity score model and the other the outcome model—we propose a corresponding semiparametric estimator based on general parametric specifications of nuisance functions. In particular, we suggest a class of parametric auto-regression models motivated by the Markov assumptions \citep{Besag1974}. By combining a restricted interference assumption with the propensity score model, we establish a new doubly robust identification result for the expected potential outcome under a hypothetical intervention on the treatment assignment vector for the entire network. We formally prove that, under weak network dependence, our proposed estimators are asymptotically normal and we characterize the impact of model misspecification on the asymptotic variance. Extensive simulation studies highlight the practical relevance of our approach. We further demonstrate its application in an empirical analysis of the NNAHRAY study, evaluating the impact of incarceration on individual socioeconomic outcomes in Brooklyn, New York.

\end{abstract}

\noindent%
\vfill

\newpage
\spacingset{1.8} 

\section{Introduction}\label{sec:intro}
Conducting causal inference on a single large network presents two fundamental challenges: (i) the data are not independently and identically distributed (non-i.i.d.), and (ii) interference is pervasive. The first challenge arises naturally in any network setting, where observations are linked through a known network structure (e.g., friendships, collaborations, or information-sharing pathways). This connectivity induces dependencies among neighboring units, thereby violating the assumption of independence. Additionally, the distribution of observables may vary across units due to differences in their network positions or roles (e.g., hubs versus peripheral nodes), leading to violations of the identical distribution assumption. The second challenge, interference, has received considerable attention in the causal inference literature. Formally, interference represents a departure from the stable unit treatment value assumption, or SUTVA \citep{cox1958planning, rubnin1972}, as it allows the outcome of one unit to be influenced by the treatment status of other units. This phenomenon gives rise to spillover effects, which are central to understanding causal relationships in many empirical applications \citep{sobel2006randomized, Rosenbaum01032007, hudgens2008, Hong2006, graham2008identifying, manski2013identification, TchetgenTchetgen2012}.

In this paper, we propose an identification and estimation strategy to address the two challenges mentioned above primarily in the context of observational studies, although our results may also be of interest in network experimental settings. Specifically, we develop a novel network version of the Augmented Inverse Propensity Weighted, which we refer to as Network AIPW (NAIPW). This approach combines the strengths of both a network version of a propensity model and a network version of an outcome model to achieve identification and estimation. We begin by establishing the causal parameter identified by our NAIPW moment equation under a standard network conditional ignorability condition. This causal parameter, denoted by $\mE[Y_i(\mathbf{a}_{\cN_i}, \mathbf{A}_{-\cN_i})]$, can be interpreted as the expected potential outcome for unit $i$ given a fixed local treatment assignment vector $\mathbf{a}_{\cN_i}$ for its network neighborhood up to a certain distance, while averaging over $\mathbf{A}_{-\cN_i}$, the treatment assignments of all other units in the network. Crucially, we note that the causal estimand in question does not a priori impose a restriction on the degree of interference; specifically, the hypothetical intervention on the unit's user-specified neighborhood $\cN_i$ does not preclude interference effects beyond the specified neighborhood and in fact remains agnostic to the prospect of such long-range spillover effects. In other words, to use the language introduced by \cite{Aronow2017}, we do not a priori assume a well-specified treatment exposure mapping. Notably, NAIPW is shown to satisfy a doubly robust identification property in that the identification result holds if either a network model for the outcome process is well-specified, or a network model for the treatment allocation mechanism is correctly specified, but not necessarily both are correct, without apriori knowledge of which model is correctly specified.

Next, to make statistical estimation and inference more tractable, we introduce two sets of Markov assumptions; one pertaining to the outcome regression model and the other to the propensity score model. These assumptions, motivated a corresponding so-called \emph{Chain Graph} models—a widely used class of graphical statistical models proposed by \cite{auto-g, shpitser2024}—assert that a unit’s observables are conditionally independent of the rest of the network given the information of its neighbors. Under the Markov assumption, we reduce the dimensionality of the outcome regression and propensity score models, yielding tractable network-based auto-regression models which we refer to as \emph{Auto-models}. The resulting auto-regression network AIPW moment function is correspondingly referred to as the \emph{Auto AIPW (AAIPW) estimating function}.

Subsequently, by further imposing a restricted interference assumption that limits interference to local units, we establish a stronger double robustness identification result. Specifically, we show that under either the restricted interference assumption or the Markov property on the outcome model, our NAIPW estimating function identifies the unit-level expected potential outcome, $\mE[Y_i(\mathbf{a})]$, where $\mathbf{a}$ denotes a fixed treatment assignment vector for all units in the network. The restricted interference assumption imposes a causal restriction that limits the influence of other units' treatments to a local neighborhood around each unit—effectively assuming a known and well-specified exposure mapping \citep{Aronow2017}. In contrast, the Markov assumption on the outcome model is a statistical restriction on the dependence structure of the outcome process. Furthermore, we demonstrate that our AAIPW function identifies $\mE[Y_i(\mathbf{a})]$ under either the outcome model's Markov property alone or the combination of restricted interference and a Markov assumption on the propensity score model.

Finally, we propose a doubly robust estimation strategy that relies on parametric models for the nuisance parameters in our AAIPW function. As a recommendation for empirical practice, we suggest using a class of parametric auto-regression models based on the Chain Graph framework \citep{Besag1974}. This parametric framework provides a coherent model for the joint distribution of the observed data over the network and are consistent with the Markov assumptions introduced earlier. Leveraging the large-network asymptotic theory recently developed by \cite{kojevnikov2021limit}, we establish the asymptotic normality of our network average AAIPW-based estimators under a large network asymptotic regime in which the network is not excessively dense, preventing any major hub of network dependence from dominating. We also characterize the impact of model misspecification on the asymptotic distribution of our estimator and obtain a variance estimator that is consistent irrespective of whether the treatment allocation mechanism or the network outcome regression model is well-specified. We evaluate the empirical performance of our proposed methods through extensive simulation studies, as well as an empirical application to the Networks, Norms and HIV/STI Risk Among Youth (NNAHRAY) study \citep{Friedman17} to evaluate the impact of incarceration on the employment and economic outcomes of residents in an impoverished neighborhood of Brooklyn, New York.

\textbf{Related Literature.} There is a growing literature on conducting causal inference with network data. Early work on interference considered settings where the population is partitioned into non-overlapping blocks, allowing for arbitrary interference and dependence within blocks but prohibiting interference and dependence across blocks. This framework is commonly referred to as \emph{partial interference} \citep{sobel2006randomized,Hong2006,hudgens2008,TchetgenTchetgen2012, Liu2014, Lundin2014, Ferracci2014, qu2024}. More recent literature has sought to relax this assumption by allowing for more general patterns of interference that do not adhere to a block structure. Instead, these models typically restrict a unit’s set of interfering units to a small set defined by spatial proximity or network connections, while also limiting the degree of outcome dependence to facilitate inference \citep{VerbitskySavitz2012,Liu2016,Aronow2017,Sofrygin2017}. A separate strand of work has primarily focused on detection of specific forms of spillover effects in the context of an experimental design in which the intervention assignment process is known to the analyst \citep{Aronow2012,vanderweele2012,Bowers2013,athey2018exact,puelz2022graph,Leung2022,gao2023}. Some recent studies in experimental settings allow for either unknown, or incompletely specified, or completely misspecified interference structures and show that a well-defined causal estimand may still be recovered under certain conditions \citep{savje2021, savje2023, weinstein2024}; potentially leveraging proxy measurements of the true interference structure \citep{weinstein2025}. In contrast, our paper considers potential interference induced by a \emph{known} network structure; for which as we establish, stronger identification and inferential results become possible.

Our paper is closely related to the recent work by \cite{auto-g}, which specifies a Chain Graph model to capture long-range dependence and accommodate general interference.\footnote{Partly inspired by \cite{auto-g}, there has been a fast-growing literature investigating causal chain graphs; see, for example, \cite{Sherman2018, Bhattacharya2020, Ogburn2020, bhattacharya2024causal, shpitser2024}.} Our estimation strategy extends their auto-outcome regression-based approach by incorporating a propensity score auto-regression model, to produce a doubly robust moment equation. Beyond the benefit of double robustness, the inclusion of a propensity model in our method facilitates the reconciliation with experimental settings where treatment allocation is known by design. In contrast to the coding estimator in \cite{auto-g}, which fails to leverage any available knowledge about the treatment allocation mechanism (as in experimental settings), and may be inefficient by discarding a substantial portion of the sample, our approach uses all available data in the network—both for the estimation of nuisance parameters and for causal effect estimation—based on the $\psi$-dependent network framework of \cite{kojevnikov2021limit}; while appropriately accounting for uncertainty in both stages.

Our paper builds on a large body of work on the use of propensity score modeling to estimate treatment effects in the presence of interference, including, but not limited to, \cite{TchetgenTchetgen2012, perez2014assessing, Liu2016, barkley2020causal, forastiere2021identification}. Our work is also closely related to recent studies that develop AIPW estimators under partial interference, such as \cite{liu2019doubly, park2020efficient}. A recent contribution by \cite{Ogburn02012024} proposes a doubly robust estimator under interference based on a causal structural equation model, assuming both a correctly specified exposure mapping and the absence of contagion. In contrast, we develop AAIPW estimators that accommodate general network interference and explicitly allow for contagion.

The remainder of this paper is organized as follows. Section~\ref{sec:setup} describes the setup and notation and presents the first identification result under a standard notion of network ignorability. Section~\ref{sec:markov} introduces the Markov property assumptions and compares them with existing approaches in the interference literature. Section~\ref{sec:main-result} presents the main identification results based on the Markov assumptions, and the proposed estimation strategies, along with asymptotic results under weak network dependence; and an extension of our framework to accommodate a user-defined exposure mapping. Section~\ref{sec:simulation} describes the results from a simulation study evaluating the performance of the proposed approach. Section~\ref{sec:empirical} illustrates the proposed methods in an empirical application based on the NNAHRAY study in \cite{Friedman17}. Finally, Section~\ref{sec:conclusion} offers some concluding remarks and directions for future research.

\section{Setup}\label{sec:setup}
Suppose we observe data on a population of $n$ interconnected units. For each unit $i \in [n]:=\{1,\dots,n\}$, we observe $(Y_i, A_i, L_i)$, where $Y_i \in \mathbb{R}$ denotes the observed outcome of interest, $A_i \in \{0, 1\}$ denotes the binary treatment received, and $L_i \in \mathbb{R}^p$ denotes a vector of observed pretreatment covariates. Let $\mathbf{A} \in \{0,1\}^n$ and $\mathbf{Y} \in \mathbb{R}^n$ denote the treatment and outcome vectors for all $n$ units, and let $\mathbf{L} \in \mathbb{R}^{n \times p}$ denote the corresponding covariate matrix. We define $\mathscr{A}(n)$ as the set of all possible treatment assignment vectors of length $n$; for example, $\mathscr{A}(2) = \{(0,0), (0,1), (1,0), (1,1)\}$. As is standard in causal inference, we assume the existence of potential outcomes $\{Y_i(\mathbf{a}) : \mathbf{a} \in \mathscr{A}(n)\}$ for each unit $i$, and define $\mathbf{Y}(\mathbf{a}) = (Y_1(\mathbf{a}), \dots, Y_n(\mathbf{a}))'$ as the vector of potential outcomes under treatment allocation $\mathbf{a} = (a_1, \dots, a_n) \in \{0,1\}^n$. The observed and potential outcomes are linked through the network version of the consistency assumption:
\begin{assumption}[Network Outcome Consistency]\label{ass:consistency}
\begin{equation*}
    \mathbf{Y} = \mathbf{Y}(\mathbf{A}) \quad \text{a.s.}
\end{equation*}
\end{assumption}
For identification, we impose a network analog of the conditional ignorability assumption:
\begin{assumption}[Network Conditional Ignorability]\label{ass:ignorability}
\begin{equation*}
    \mathbf{A} \independent \mathbf{Y}(\mathbf{a}) \mid \mathbf{L} \quad \text{for all } \mathbf{a} \in \mathscr{A}(n)~.
\end{equation*}
\end{assumption}
This assumption, also known as conditional exchangeability or unconfoundedness, states that all relevant information used to generate the treatment assignment—whether by a researcher in an experiment or by nature in an observational setting—is captured by the covariates $\mathbf{L}$. In contrast, in the i.i.d. setting, the commonly imposed ignorability assumption is marginal and applies at the unit level: $A_i \independent Y_i(a_i) \mid L_i$.

Finally, we make the following positivity assumption at the network treatment allocation level: 
\begin{assumption}[Network Positivity]\label{ass:positivity}
    \begin{equation*}
        f(\mathbf{a}\mid \mathbf{L}) > \sigma > 0 \text{ a.e. for all } \mathbf{a} \in \mathscr{A}(n)~.
    \end{equation*}
\end{assumption}
We consider unit-level potential outcome expectations $\mE\left[ Y_i(\mathbf{a}) \right]$ for $i \in [n] $ as one of our primary identification targets. These expectations serve as fundamental building blocks for a wide array of causal estimands commonly studied in the literature. Following the framework of \citet{hudgens2008} and \citet{auto-g}, we define the direct and indirect (or spillover) effects as follows:
\begin{align*}
    \DE_i(\mathbf{a}_{-i}) &= \mE\left[ Y_i(\mathbf{a}_{-i}, a_i = 1) \right] - \mE\left[ Y_i(\mathbf{a}_{-i}, a_i = 0) \right]~, \\
    \IE_i(\mathbf{a}_{-i}) &= \mE\left[ Y_i(\mathbf{a}_{-i}, a_i = 0) \right] - \mE\left[ Y_i(\mathbf{a}_{-i} = \mathbf{0}, a_i = 0) \right]~,
\end{align*}
where $\mathbf{a}_{-i}$ denotes the treatment assignments for all units other than unit~$i$. These effects can be further averaged over a hypothetical allocation regime $\pi_i(\mathbf{a}_{-i}; \alpha)$ indexed by $\alpha$ to obtain allocation-specific unit-level average direct and spillover effects: $\DE_i(\alpha) = \sum_{\mathbf{a}_{-i} \in \mathscr{A}(n-1)} \DE_i(\mathbf{a}_{-i}) \, \pi_i(\mathbf{a}_{-i}; \alpha)$ and $\IE_i(\alpha) = \sum_{\mathbf{a}_{-i} \in \mathscr{A}(n-1)} \IE_i(\mathbf{a}_{-i}) \, \pi_i(\mathbf{a}_{-i}; \alpha)$. Finally, averaging over all units in the network yields the allocation-specific population average direct and spillover effects: $\DE(\alpha) = n^{-1} \sum_{i=1}^n \DE_i(\alpha)$ and $\IE(\alpha) = n^{-1} \sum_{i=1}^n \IE_i(\alpha)$. Another parameter of interest in our paper is $\mE[Y_i(\mathbf{a}_{\cN_i}, \mathbf{A}_{-\cN_i})]$, for which direct and indirect effects can be defined analogously.  As we will show later, identification of $\mE[Y_i(\mathbf{a})]$ requires stronger assumptions than $\mE[Y_i(\mathbf{a}_{\cN_i}, \mathbf{A}_{-\cN_i})]$.

We now introduce additional notation used in our identification strategy for general network data under interference. Let $\mathcal E_n$ denote the set of edges in the network where each element $(i, j) \in \mathcal{E}_n$ indicates that units $i$ and $j$ are directly connected. The set $\mathcal{E}_n$ thus characterizes the network by specifying which pairs of units are neighbors and implies a distance function, denoted by $d_n(i, j)$, which is the distance between units $i$ and $j$ on the network.\footnote{The distance between two vertices in a graph is the number of edges in a shortest path (also called a graph geodesic) connecting them.} Let $\cN_n^\partial(i; s) = \{j \in [n] : d_n(i, j) = s\}$ denote the $s$-step neighborhood boundary of unit $i$ and $\cN_n(i; s) = \{j \in [n] : d_n(i, j) \leq s\}$ denote the $s$-step neighborhood of $i$, which includes all units within $s$ steps of $i$, including $i$ itself. Let $\mathbf{A}_{\cN_n(i, s)}$ denote the treatment assignment vector for unit $i$ and all units within its $s$-step neighborhood, and $\mathbf{A}_{-\cN_n(i, s)}$ denote the assignment vector for all remaining units outside this neighborhood. For notation simplicity, we write $\mathbf{A}_{\cN_i} = \mathbf{A}_{\cN_n(i, K)}$ and $\mathbf{A}_{-\cN_i} = \mathbf{A}_{-\cN_n(i, K)}$ for a fixed distance $K \ge 1$, when there is no risk of ambiguity. Finally, let $\mathbf{Y}_{-i}$ denote the vector of observed outcomes for all units in the network except unit~$i$.

At this stage, we impose the network versions of the three standard causal identification assumptions: consistency, positivity, and conditional ignorability. The following theorem presents our causal identification result in the form of a moment condition under these assumptions, which we refer to as Network AIPW (NAIPW).

\begin{theorem}\label{thm:identification1}
    Under Assumption \ref{ass:consistency}-\ref{ass:positivity}, for any fixed distance $K\ge 1$ we have:
    \begin{equation}\label{eqn:identification1}
    \begin{aligned}
        &\mE\left[  \frac{1\left\{ \mathbf{A}_{\cN_i} = \mathbf{a}_{\cN_i} \right\} \left(Y_i - \mE[Y_i \mid \mathbf{Y}_{-i}, \mathbf{A}_{\cN_i}=\mathbf{a}_{\cN_i}, \mathbf{A}_{-\cN_i}, \mathbf{L}]\right)}{\mP\left( \mathbf{A}_{\cN_i} = \mathbf{a}_{\cN_i} \mid \mathbf{L}, \mathbf{A}_{-\cN_i} \right)} + \mE[Y_i \mid \mathbf{Y}_{-i}, \mathbf{A}_{\cN_i}=\mathbf{a}_{\cN_i}, \mathbf{A}_{-\cN_i}, \mathbf{L}] \right] \\
        &= \mE\left[ Y_i(\mathbf{a}_{\cN_i}, \mathbf{A}_{-\cN_i}) \right]~,
    \end{aligned}
    \end{equation}
    if either $\mP\left( \mathbf{A}_{\cN_i} = \mathbf{a}_{\cN_i} \mid \mathbf{L}, \mathbf{A}_{-\cN_i} \right)$ or $\mE[Y_i \mid \mathbf{Y}_{-i}, \mathbf{A}_{\cN_i}=\mathbf{a}_{\cN_i}, \mathbf{A}_{-\cN_i}, \mathbf{L}]$ is correctly specified.
\end{theorem}

Theorem~\ref{thm:identification1} establishes identification of the average potential outcome $\mE\left[ Y_i(\mathbf{a}_{\cN_i}, \mathbf{A}_{-\cN_i}) \right]$ through a doubly robust moment equation, under a standard network ignorability condition, along with consistency and positivity; a key result which appears to be new to the interference literature.

 In the literature on doubly robust or AIPW inference under i.i.d. sampling,  
 an i.i.d. version of the term $\mathrm{P}\left( \mathbf{A}_{\cN_i} = \mathbf{a}_{\cN_i} \mid \mathbf{L}, \mathbf{A}_{-\cN_i} \right)$ is commonly referred to as the propensity score or propensity model, and a corresponding i.i.d. version of $\mE[Y_i \mid \mathbf{Y}_{-i}, \mathbf{A}_{\cN_i}=\mathbf{a}_{\cN_i}, \mathbf{A}_{-\cN_i}, \mathbf{L}]$ is referred to as the outcome model. Based on this, we will refer to these models in the network setting as network propensity score and network outcome regression models, respectively. 
 Importantly, this identification result is doubly robust in the sense that \eqref{eqn:identification1} holds as long as either the network propensity model or the network outcome regression model is correctly specified.

Theorem \ref{thm:identification1} may be viewed as a generalization of the network g-formula of \cite{auto-g}, although there are fundamental conceptual differences. First, the counterfactual quantity targeted by \cite{auto-g}, $\mE\left[ Y_i(\mathbf{a}) \right]$, considers a global intervention on the treatment of all units over the network, while the target of inference in Theorem \ref{thm:identification1} identifies the counterfactual mean $\mE\left[ Y_i(\mathbf{a}_{\cN_i}, \mathbf{A}_{-\cN_i}) \right]$ under a more local intervention over the treatment allocation of neighboring units, averaging over the observed marginal distribution of the treatment of all other units in the network.\footnote{\label{fn:weak-identification}Notably, while the global causal parameter $\mE\left[ Y_i(\mathbf{a}) \right]$ remains identified under network consistency, positivity and the relatively weaker network ignorability condition $\mathbf{A} \independent Y_i(\mathbf{a}) \mid \mathbf{L}   \text{ for all }   \mathbf{a} \in \mathscr{A}(n),$ via the auto-g-computation algorithm of \cite{auto-g}, we show in Section \ref{app:extra-identification} of the Appendix that under the weaker ignorability condition, the local parameter $\mE\left[ Y_i(\mathbf{a}_{\cN_i}, \mathbf{A}_{-\cN_i}) \right]$ is identified by the NAIPW only if the network propensity score is correctly specified and the outcome regression is set to a function which does not depend on $\mathbf{Y}_{-i}$, and not otherwise.} In experimental settings where the treatment allocation mechanism and therefore inverse probability weight is known by design, by double robustness, the identifying moment equation readily provides an unbiased estimator for the target parameter even if the unknown network outcome regression is set to an arbitrary/incorrect value.  In contrast, the identification result in itself may be of limited value in observational settings where neither nuisance function is known apriori, as they are technically not statistically identified from the observed data without an additional restriction, due to the fundamental limitation of having observed a single realization of a network of heterogeneous dependent units. 

A related observation on the fundamental limits of statistical identification in network-based counterfactual inference was made by \cite{auto-g}, who motivated the use of a Chain Graph model as a path forward for statistical inference via the so-called auto-g-computation algorithm, a network generalization of Robins' foundational g-computation algorithm. A conceptually analogous strategy will be adopted in the sections that follow. However, as we will show, our proposed framework—which focuses on identification through Markov assumptions motivated by the Chain Graph model—offers several robustness advantages over auto-g-computation.

\section{The Markov Properties for Network Data}\label{sec:markov}
The identification result in Theorem~\ref{thm:identification1} highlights a doubly robust moment condition but does not immediately yield a feasible estimation strategy. In particular, both the network propensity score and outcome regression terms depend on information from the entire network, such as the full covariate vector $\mathbf{L}$ and the complete treatment assignment vector $\mathbf{A}$. To address this challenge,  we introduce a set of \textit{Markov property} assumptions that underlie our identification strategies. These assumptions allow us to reduce the dimensionality of the outcome regression and propensity score models, making them tractable for estimation.

\begin{assumption}[Markov Property]\label{ass:markov}
    \quad
    \begin{enumerate}
        \item[(i)] The outcome process $\mathbf{Y} \mid \mathbf{A}, \mathbf{L}$ satisfies:
        \begin{align*}
            Y_i \independent \{Y_k, A_k, L_k\} \mid \left(A_i, L_i, \{Y_j, A_j, L_j : j \in \cN_n^\partial(i;1)\}\right) \text{ for all } i \text{ and } k \notin \cN_n(i;1)~.
        \end{align*}
        
        \item[(ii)] The treatment assignment process $\mathbf{A} \mid \mathbf{L}$ satisfies:
        \begin{align*}
            \mathbf{A}_{\cN_n(i;\ell)} &\independent \{A_k, L_k\} \mid \left(\mathbf{L}_{\cN_n(i;\ell)}, \{A_j, L_j : j \in \cN_n^\partial(i;\ell+1)\}\right)
        \end{align*}
        for all $i$ and $k\notin \cN_n(i;\ell+1)$ and $\ell \geq 0$.
    \end{enumerate}
\end{assumption}
The Markov properties in Assumption~\ref{ass:markov} are motivated by a class of graphical statistical models for network data, commonly referred to as Chain Graph models \citep{lauritzen2002chain}. Their global Markov properties naturally imply the specific conditional independence relations stated in Assumption~\ref{ass:markov}. (See Appendix~\ref{app:chain-graph} for additional background on Chain Graph models and their global Markov properties.) These models provide formal justification for both the conditional independence assumptions and the use of parametric auto-regression models, which we later recommend as a natural approach for estimating nuisance functions. Importantly, these Markov assumptions allow for both long-range dependence in the data and long-range interference, while also enabling localization of the nuisance functions—thereby making the estimation problem more tractable despite the presence of complex network dependencies. Moreover, these assumptions yield testable implications: unconditional moment conditions can be derived from the conditional independence restrictions, enabling the construction of overidentification tests.

\section{Main Results}\label{sec:main-result}
In this section, we present our main identification results. Following identification, we then introduce an AAIPW estimator and establish its consistency and asymptotic normality. Our asymptotic results rely on recent developments in the theory of weak dependence in networks, specifically the framework developed in \cite{kojevnikov2021limit}.

\subsection{Identification with Markov Assumptions}
We begin with the following proposition, which mitigates the high-dimensional conditioning problem by localizing the relevant information sets using Markov assumptions.
\begin{proposition}\label{prop:models}
    Under Assumption~\ref{ass:markov}(i), we have
    \begin{equation*}
        \mE[Y_i \mid \mathbf{Y}_{-i}, \mathbf{A} = \mathbf{a}, \mathbf{L}] = \mE[Y_i \mid \mathbf{A}_{\cN_i} = \mathbf{a}_{\cN_i}, L_i, \{Y_j, L_j : j \in \cN^\partial_n(i;1)\}]~.
    \end{equation*}
    Under Assumption~\ref{ass:markov}(ii), for any $K \geq 1$, we have
    \begin{equation*}
        \mP\left(\mathbf{A}_{\cN_i} = \mathbf{a}_{\cN_i} \mid \mathbf{L}, \mathbf{A}_{-\cN_i}\right) = \mP(\mathbf{A}_{\cN_i} = \mathbf{a}_{\cN_i} \mid \{L_j : j \in \cN_n(i, K+1)\}, \{A_j : j \in \cN_n^\partial(i; K+1)\})~.
    \end{equation*}
\end{proposition}

To simplify notation, we introduce the following definitions for the nuisance functions:
\begin{align}
    \beta_{\mathbf{a}_{\cN_i}}(X_i^y) &= \mE[Y_i \mid \mathbf{A}_{\cN_i} = \mathbf{a}_{\cN_i}, L_i, \{Y_j, L_j : j \in \mathcal{N}^\partial_n(i;1)\}]~, \label{eqn:nuisance-function1}\\
    \pi_{\mathbf{a}_{\cN_i}}(X_i^a) &= \mP(\mathbf{A}_{\cN_i} = \mathbf{a}_{\cN_i} \mid \{L_j : j \in \cN_n(i, K+1)\}, \{A_j : j \in \mathcal{N}^\partial_n(i;K+1)\})~.\label{eqn:nuisance-function2}
\end{align}
$X_i^y = (L_i, (Y_j, L_j: j\in \mathcal{N}^\partial_n(i;1)))$ and $X_i^a = ((L_j: j \in \mathcal N_n(i,K+1)), (A_j:j\in \mathcal N_n^\partial(i;K+1) ))$ represent the ``covariates'' used in the outcome and propensity models, respectively. Following the terminology introduced in \cite{auto-g}, we refer to $\beta_{\mathbf{a}_{\cN_i}}(X_i^y)$ and $\pi_{\mathbf{a}_{\cN_i}}(X_i^a)$ as \emph{Auto-functions}, reflecting their interpretation as network-based auto-regression models. For example, the outcome Auto-function involves regressing unit $i$'s outcome on its neighbors’ outcomes, which serves as a spatial analogue of ``lags'' defined by network distance. We refer to the network AIPW (NAIPW) function with Auto-functions substituted in as the Auto-AIPW (AAIPW) function. The following theorem establishes an identification result based on the AAIPW function.

\begin{proposition}[Identification for AAIPW Estimation (local intervention)]\label{thm:identification1-markov}
Suppose Assumptions \ref{ass:consistency}--\ref{ass:positivity} hold. Then the following identification result holds:
\begin{equation*}
    \mE\left[  \frac{1\{ \mathbf{A}_{\cN_i} = \mathbf{a}_{\cN_i} \} }{\pi_{\mathbf{a}_{\cN_i}}(X_i^a)} \left(Y_i - \beta_{\mathbf{a}_{\cN_i}}(X_i^y)\right) + \beta_{\mathbf{a}_{\cN_i}}(X_i^y) \right] = \mE\left[ Y_i(\mathbf{a}_{\cN_i}, \mathbf{A}_{-\cN_i}) \right]~.
\end{equation*}
provided that either Assumption~\ref{ass:markov}(i) or Assumption~\ref{ass:markov}(ii) holds.
\end{proposition}

\subsection{Identification for Global Intervention}\label{subsec:main1}
To strengthen the identification result in Theorem~\ref{thm:identification1}, we now consider conditions under which the average potential outcome under global interventions, $\mE[Y_i(\mathbf{a})]$, can be identified. We first consider the following restricted interference assumption.
\begin{assumption}[Restricted Interference]\label{ass:interference}
    For all units $i \in [n]$ and any two treatment assignments $\mathbf{a}, \mathbf{a}^\prime \in \mathscr{A}(n)$,
    \begin{equation*}
        Y_i(\mathbf{a}) = Y_i(\mathbf{a}^\prime) \quad \text{w.p.1 if } \mathbf{a}_{\cN_i} = \mathbf{a}^\prime_{\cN_i}.
    \end{equation*}
\end{assumption}
This assumption defines the scope of interference in the network. It asserts that each unit’s potential outcome depends only on the treatment assignments within its $K$-step neighborhood $\cN_i$, and is invariant to assignments outside this set. 


Under the restricted interference assumption, it is immediate that the global parameter coincides with the local one. Perhaps more surprisingly, Assumption~\ref{ass:markov}(i)—which imposes a Markov property on the outcome process—also suffices for identification. The following theorem presents the strengthened identification result, showing that either the restricted interference assumption or the outcome model’s Markov property alone is sufficient to identify $\mE[Y_i(\mathbf{a})]$.

\begin{theorem}[Strengthened Identification]\label{thm:identification2}
    Suppose Assumption~\ref{ass:consistency}-\ref{ass:positivity} holds. Then, under either Assumption~\ref{ass:interference} or Assumption~\ref{ass:markov}(i), for any $K \geq 1$, we have:
    \begin{equation}\label{eqn:identification2}
    \begin{aligned}
        &\mE\left[  \frac{1\left\{ \mathbf{A}_{\cN_i} = \mathbf{a}_{\cN_i} \right\} \left(Y_i - \mE[Y_i \mid \mathbf{Y}_{-i}, \mathbf{A}_{\cN_i}=\mathbf{a}_{\cN_i}, \mathbf{A}_{-\cN_i}, \mathbf{L}]\right)}{\mP\left( \mathbf{A}_{\cN_i} = \mathbf{a}_{\cN_i} \mid \mathbf{L}, \mathbf{A}_{-\cN_i} \right)} + \mE[Y_i \mid \mathbf{Y}_{-i}, \mathbf{A}_{\cN_i}=\mathbf{a}_{\cN_i}, \mathbf{A}_{-\cN_i}, \mathbf{L}] \right]\\
        &= \mE\left[ Y_i(\mathbf{a}) \right]~.
    \end{aligned}
    \end{equation}
\end{theorem}
Theorem~\ref{thm:identification2} establishes a double robustness property in identification of $\mE[Y_i(\mathbf{a})]$: if either the restricted interference assumption (Assumption~\ref{ass:interference}) or the Markov assumption for the outcome (Assumption~\ref{ass:markov}(i)) holds, the same NAIPW function identifies the average potential outcome $\mE[Y_i(\mathbf{a})]$. In this sense, the moment condition is robust to misspecification of one of the two assumptions.

\begin{remark}
    It is straightforward to see why Assumption~\ref{ass:interference} suffices: under this restricted interference assumption, the treatment assignments outside of $\cN_i$ do not affect $Y_i$, so that $\mE[Y_i(\mathbf{a})]  = \mE[Y_i(\mathbf{a}_{\cN_i}, \mathbf{A}_{-\cN_i})]$. What is less immediate—and perhaps more interesting—is that the Markov assumption on the outcome also implies that $\mE[Y_i(\mathbf{a}_{\cN_i}, \mathbf{A}_{-\cN_i})] = \mE[Y_i(\mathbf{a})]$. This illustrates that the Markov property on the conditional outcome distribution can be interpreted as enforcing a form of restricted interference, making the two target parameters in Theorem \ref{thm:identification1} and \ref{thm:identification2} equivalent. In other words, both assumptions ultimately serve to effectively reduce the dimension of the interference set.
\end{remark}

\begin{remark}
    A further insight from Theorem~\ref{thm:identification2} is that the NAIPW function combines the strengths of both models without compromising the identification that would be attainable under either model alone. To illustrate this, note that under Assumption~\ref{ass:interference}, the inverse probability weighting (IPW) component alone identifies $\mE[Y_i(\mathbf{a})]$ using standard arguments based on iterated expectations and conditional ignorability:
    \begin{align*}
        \mE\left[\frac{1\left\{ \mathbf{A}_{\cN_i} = \mathbf{a}_{\cN_i} \right\} Y_i}{\mP\left( \mathbf{A}_{\cN_i} = \mathbf{a}_{\cN_i} \mid \mathbf{L}, \mathbf{A}_{-\cN_i} \right)}\right] 
        &= \mE\left[\mE\left[\frac{1\left\{ \mathbf{A}_{\cN_i} = \mathbf{a}_{\cN_i} \right\} Y_i(\mathbf{a}_{\cN_i})}{\mP\left( \mathbf{A}_{\cN_i} = \mathbf{a}_{\cN_i} \mid \mathbf{L}, \mathbf{A}_{-\cN_i} \right)} \mid \mathbf{L}, \mathbf{A}_{-\cN_i}\right]\right] \\
        &= \mE\left[\mE[Y_i(\mathbf{a}_{\cN_i}) \mid \mathbf{L}, \mathbf{A}_{-\cN_i}]\right] = \mE[Y_i(\mathbf{a}_{\cN_i})] = \mE[Y_i(\mathbf{a})]~.
    \end{align*}
    Similarly, under Assumption~\ref{ass:markov}(i), the outcome regression term alone identifies $\mE[Y_i(\mathbf{a})]$:
    \begin{align*}
        \mE\left[\mE[Y_i \mid \mathbf{Y}_{-i}, \mathbf{A}_{\cN_i}=\mathbf{a}_{\cN_i}, \mathbf{A}_{-\cN_i}, \mathbf{L}]\right] &=\mE\left[\mE[Y_i \mid \mathbf{Y}_{-i}, \mathbf{A} = \mathbf{a}, \mathbf{L}]\right] \\
        &= \mE\left[\mE[Y_i(\mathbf{a}) \mid \mathbf{Y}_{-i}(\mathbf{a}), \mathbf{A} = \mathbf{a}, \mathbf{L}]\right] \\
        &= \mE\left[\mE[Y_i(\mathbf{a}) \mid \mathbf{Y}_{-i}(\mathbf{a}), \mathbf{L}]\right] = \mE[Y_i(\mathbf{a})]~.
    \end{align*}
    It is therefore notable that a single moment function suffices to achieve the same identification result as would be obtained under either the restricted interference or the outcome regression model alone, requiring only one of the two assumptions to hold.
\end{remark}

\begin{remark}
    It is worth noting that Assumption~\ref{ass:markov}(i) can be replaced by a weaker condition stated in terms of the functional form of the conditional expectation:
    \begin{equation}\label{eqn:remark-outcome-model-condexp}
        \mE[Y_i \mid \mathbf{Y}_{-i}, \mathbf{A}=\mathbf{a}, \mathbf{L}] = \mE[Y_i \mid \mathbf{Y}_{-i}, \mathbf{A}_{\cN_i}=\mathbf{a}_{\cN_i}, \mathbf{L}]~.
    \end{equation}
    This condition requires that the conditional outcome function depends only on the treatment assignments within the $K$-degree neighborhood.
\end{remark}

Similar to Theorem~\ref{thm:identification1}, the identification result in Theorem~\ref{thm:identification2} does not immediately yield a feasible estimation strategy due to high dimensionality of nuisance functions. By incorporating the Markov assumption in the propensity model, we obtain the following identification result for estimation.

\begin{proposition}[Identification for AAIPW Estimation (global intervention)]\label{thm:identification2-markov}
Suppose Assumptions \ref{ass:consistency}--\ref{ass:positivity} hold. Then the following identification result holds:
\begin{equation*}
    \mE\left[  \frac{1\{ \mathbf{A}_{\cN_i} = \mathbf{a}_{\cN_i} \} }{\pi_{\mathbf{a}_{\cN_i}}(X_i^a)} \left(Y_i - \beta_{\mathbf{a}_{\cN_i}}(X_i^y)\right) + \beta_{\mathbf{a}_{\cN_i}}(X_i^y) \right] = \mE\left[ Y_i(\mathbf{a}) \right]~.
\end{equation*}
provided that one of the following two conditions holds:
\begin{enumerate}
    \item Assumption~\ref{ass:markov}(i) holds; or
    \item Assumptions~\ref{ass:markov}(ii) and~\ref{ass:interference} hold
\end{enumerate}
\end{proposition}

\subsection{AAIPW Estimators}\label{subsec:main2}
In this section, we propose a feasible AAIPW estimator based on our identification results using AAIPW functions. Specifically, we adopt a parametric specification of the Auto-functions and leave the nonparametric approach for future work. To begin with, we assume the following parametric specifications for Auto-functions defined in \eqref{eqn:nuisance-function1} and \eqref{eqn:nuisance-function2}.

\begin{assumption}\label{ass:param}
Auto-functions $\beta_{\mathbf{a}_{\cN_i}}(X_i^y)$ and $\pi_{\mathbf{a}_{\cN_i}}(X_i^a)$ satisfy
\begin{itemize}
    \item[(i)]$\beta_{\mathbf{a}_{\cN_i}}(X_i^y)$ can be parameterized by $\theta_0$ in the sense that  $ \beta_{\mathbf{a}_{\cN_i}}(X_i^y) = \beta_{\mathbf{a}_{\cN_i}}(X_i^y;\theta_0)$ for some finite-dimensional $\theta_0$.
    \item[(ii)] $\pi_{\mathbf{a}_{\cN_i}}(X_i^a)$ can be parameterized by $\eta_0$ in the sense that $\pi_{\mathbf{a}_{\cN_i}}(X_i^a) = \pi_{\mathbf{a}_{\cN_i}}(X_i^a;\eta_0)$ for some finite-dimensional $\eta_0$.
\end{itemize}
\end{assumption}

Let $\beta_{\mathbf{a}_{\cN_i}}(X_i^y;\widehat\theta)$ and $\pi_{\mathbf{a}_{\cN_i}}(X_i^a;\widehat\eta)$ denote the estimated outcome and propensity model based on the parametric models. Specifically,  consider estimating the outcome with the method of moments. Note that
\begin{equation*}
    \mE[Y_i - \beta_{\mathbf{a}_{\cN_i}}(X_i^y;\theta) \mid X_i^y] = 0 \Longrightarrow \mE[(Y_i - \beta_{\mathbf{a}_{\cN_i}}(X_i^y;\theta) )h(X_i^y;\theta)] = 0
\end{equation*}
where $h$ is a function to be specified by the user and has the same dimension as $\theta$ so that we have exact identification. Define $g^y(Y_i, X_i^y;\theta) = (Y_i - \beta_{\mathbf{a}_{\cN_i}}(X_i^y;\theta) )h(X_i^y;\theta)$, then we have moment conditions $\mE[g^y(Y_i),X_i^y;\theta]=0$ for the outcome model and our outcome model parameter $\widehat\theta$ can be estimated by solving the sample moment equations: $\frac{1}{n}\sum_{i=1}^n g^y(Y_i, X_i^y;\widehat\theta)=0$. Similarly, $\widehat\eta$ can be solved based on some moment conditions $\mE[g^a(\mathbf{A}_{\cN_i}, X_i^a; \eta)]=0$. We then make the following assumptions about the estimator:
\begin{assumption}\label{ass:estimator_consistent}
The estimators for the outcome model and the propensity score model converge in probability to fixed (possibly misspecified) limits: $\widehat\theta \xrightarrow{p} \theta$ and $\widehat\eta \xrightarrow{p} \eta$ for some parameters $\theta$ and $\eta$.
\end{assumption}

This assumption states that the estimators stabilize asymptotically, converging to well-defined probability limits.
Notably, we do not assume that the probability limits coincide with the true  parameters of either the outcome model or the propensity model. If Assumption \ref{ass:param}(i) or \ref{ass:param}(ii) holds, and we have $\theta = \theta_0$ or $\eta = \eta_0$, then we ensure convergence to the desired outcome or propensity model in probability, respectively. Consider the case in which the propensity score model is correctly specified while the outcome model is misspecified, Assumption \ref{ass:estimator_consistent} still assumes that the resulting estimator for the outcome, $\widehat\theta$, converges to some fixed value $\theta$ and the following properties hold:
\begin{align*}
     \beta_{\mathbf{a}_{\cN_i}}(X_i^y;\theta) &\neq \mE[Y_i \mid \mathbf{A}_{\cN_i} = \mathbf{a}_{\cN_i}, L_i, (Y_j, L_j: j\in \mathcal{N}^\partial_n(i;1))]~; \\
     \pi_{\mathbf{a}_{\cN_i}}(X_i^a;\eta) &= \mP(\mathbf{A}_{\cN_i} = \mathbf{a}_{\cN_i}\mid (L_j: j \in \cN_n(i,2)), (A_j:j\in \mathcal N_n^\partial(i;K+1)) ~.
\end{align*}
The reverse case, where the outcome model is correctly specified and the propensity model is misspecified, can be formulated analogously. When the parameters are estimated via the generalized method of moments, Proposition \ref{prop:gmm-consistency} in the Appendix verifies this convergence under the weak network-dependence conditions introduced in Section \ref{subsec:main3clt}. Therefore, in the GMM setting, the pseudo-likelihood moment conditions guarantee consistency even when the models are misspecified, as allowed in Assumption \ref{ass:estimator_consistent}. With such parametric specifciation, we are ready to introduce our estimator based on the AAIPW function. Specifically, we construct
\[
\widehat \mu_{\mathrm{AIPW}} = \frac{1}{n} \sum_{i=1}^n \widehat W_i,
\]
where
\begin{equation}\label{eqn:AIPW}
    \widehat W_i = \frac{1\{ \mathbf{A}_{\cN_i} = \mathbf{a}_{\cN_i} \}}{\pi_{\mathbf{a}_{\cN_i}}(X_i^a; \widehat\eta)} \left(Y_i - \beta_{\mathbf{a}_{\cN_i}}(X_i^y; \widehat\theta)\right) + \beta_{\mathbf{a}_{\cN_i}}(X_i^y; \widehat\theta)~,
\end{equation}
and $\pi_{\mathbf{a}_{\cN_i}}(X_i^a; \widehat\eta)$ and $\beta_{\mathbf{a}_{\cN_i}}(X_i^y; \widehat\theta)$ denote the estimated network propensity score and outcome regression, respectively.

Although we consider general parametric specifications in this section, we present a specific modeling approach in Appendix~\ref{subsec:chain-graph-parametric}. This approach is based on the \emph{chain graph} model—a graphical statistical framework introduced by \cite{auto-g, shpitser2024}. A key advantage of the chain graph framework in the context of network data is that it provides a coherent specification for the joint distribution of observed variables across the network. In contrast, specifying parametric or nonparametric models for the marginal distributions of the outcome and treatment processes may not induce a compatible joint distribution, thereby limiting their validity in network settings.

\subsection{Asymptotic Property of the AAIPW Estimator}\label{subsec:main3clt}
In this section, we establish the asymptotic properties of the AAIPW estimator, which relies on the weak dependence assumption on the network data. Specifically, we adopt the $\psi$-network dependence framework in \cite{kojevnikov2021limit} to characterize weak dependence in network settings and assume that $Z_i = (Y_i,A_i,L_i')'$ exhibits $\psi$-network dependence with appropriate weak dependent coefficients. This ensures that the dependence structure of $W_i$ on the network is sufficiently weak for the law of large numbers and central limit theorems to hold for the AAIPW estimator. To formalize this, we first introduce the definition of $\psi$-dependence network along with the necessary notation. 

Let $\cP_n(q_1, q_2;s)$ denote the collection of two sets of nodes of size $q_1$ and $q_2$ with distance between each other of at least $s$. Formally, 
$\cP_n(q_1,q_2;s) = \{(Q_1,Q_2): Q_1, Q_2 \subset [n], |Q_1| =q_1, |Q_2|=q_2, \mbox{ and  } d_n(Q_1,Q_2) \geq s\}.$ 
In addition, $\mathcal{L}_{v,q_1}$ and $\mathcal{L}_{v,q_2}$ denote the collection of bounded Lipschitz real functions on $\mathbb{R}^{v\times q_1}$ and $\mathbb{R}^{v\times q_2},$ respectively. Using this notation, we can now state the definition of $\psi$-network dependence as follows:
\begin{definition}[$\psi-$Network Dependence]
  The triangular array $\{Z_{i}\}_{i \in [n]}, Z_{i} \in \mathbb{R}^v$, is called conditionally $\psi$-dependent if there exist a uniformly bounded sequence
  $\zeta_n = \{\zeta_{n,s}\}, s \geq 0, \zeta_{n,0} = 1$, and a collection of functionals, 
  $(\psi_{q_1,q_2})_{q_1,q_2\in\mathbb{N}}, \ \psi_{q_1,q_2}: \mathcal{L}_{v,q_1}
  \times \mathcal{L}_{v,q_2} \rightarrow [0, \infty)$, such that for all
  $(Q_1, Q_2) \in \cP_n(q_1, q_2;s)$ with $s > 0$ and all $f_1 \in
  \mathcal{L}_{v,q_1}$ and $f_2 \in
  \mathcal{L}_{v,q_2},$
  \begin{equation}
    |\Cov(f_1(\mathbf{Z}_{Q_1}), f_2(\mathbf{Z}_{Q_2}))| \leq \psi_{q_1,q_2}(f_1,f_2) \zeta_{n,s}  \ \ \  \mbox{a.s.}
  \end{equation}
\end{definition}
Then we are ready to introduce the assumptions we need. 
\begin{assumption}\label{ass:network_weak_dependence}
  The array $Z_{i}=(Y_i,A_i,L_i')'$ is conditionally $\psi$-weakly dependent with the coefficients $\zeta_{n,s}$ and functionals $\psi_{q_1,q_2}$ and there exist cosntants $\lambda_1 > 4$, $\lambda_2 > 3/2$, $\lambda_3>0$ and a sequence $\xi_n \rightarrow \infty$ such that $Z_i$, $\zeta_n$ and $\psi_{q_1,q_2}$ satisfy:
  \begin{itemize}
  \item[(i)] $\sup_i\mE[Z_i^{\lambda_1}]<\infty$,
  \item[(ii)] $\sup_{n \geq 1} \max_{s \geq 1}
  \zeta_{n,s} < \infty$ and $\zeta_{n, \xi_n}^{1-1/\lambda_1}n^{\lambda_2} = o_{a.s.}(1)$, 
  \item[(iii)] $\psi_{q_1,q_2}(f_1,f_2)\leq \lambda_3 \times q_1q_2 (||f_1||_{\infty} + \mbox{\normalfont
    Lip}(f_1))
  (||f_2||_{\infty} + \mbox{\normalfont
    Lip}(f_2))$,
  \item[(iv)] for each $k \in \{1, 2\},$
    \begin{equation*}
      \cfrac{1}{n^{k/2}} \sum_{s \geq 0} r_n(s, \xi_n; k) \zeta_{n,s}^{1
        - (k+2)/\lambda_1} = o_{a.s.}(1)
    \end{equation*}
    where we define the neighborhood shell: 
    \begin{equation*}
      r_n (s, \xi_n; k) = \inf_{\alpha > 1} \biggl\{\frac{1}{n} \sum_{i \in N_n} \max_{j \in
        \cN_n^\partial(i; s)} |\cN_n(i; \xi_n)
      \setminus \cN_n(j;s-1)|^{k\alpha}\biggr\}^{\frac{1}{\alpha}}
      \times \biggl\{ \frac{1}{n} \sum_{i \in
        N_n} |\cN_n^\partial(i;s)|^{\frac{\alpha}{\alpha-1}}\biggr\}^{1- \frac{1}{\alpha}}~,
    \end{equation*}
 where we use $|\cdot|$ to denote the cardinality of a set.
  \end{itemize}  
\end{assumption}
 $\zeta_{n,s}$ captures how fast network dependence between units decays and Assumption~\ref{ass:network_weak_dependence}(i) ensures the rate $\zeta_{n,\xi_n}$ goes to zero. $r_n(s,\xi_n;k)$ in Assumption~\ref{ass:network_weak_dependence}(iv) measures how dense the network is. Specifically, the two components of $r_n(s,\xi_n;k)$ in the curly brackets capture the denseness of the network through the average neighborhood shell size and average neighborhood sizes, respectively.

Before presenting the main central limit theorem (CLT) result, we first introduce the notation for the network target parameter of our AAIPW estimator:
\begin{equation*}
    \mu_n = \frac{1}{n}\sum_{i=1}^n \mu_i = \frac{1}{n}\sum_{i=1}^n \mE\left[  \frac{1\{ \mathbf{A}_{\cN_i} = \mathbf{a}_{\cN_i} \} }{\pi_{\mathbf{a}_{\cN_i}}(X_i^a;\eta)} \left(Y_i - \beta_{\mathbf{a}_{\cN_i}}(X_i^y;\theta)\right) + \beta_{\mathbf{a}_{\cN_i}}(X_i^y;\theta) \right] ~.
\end{equation*}
As established in Theorem~\ref{thm:identification1-markov} and \ref{thm:identification2-markov}, $\mu_n$ corresponds to a well-defined causal parameter. However, the interpretation of this quantity may differ depending on whether the restricted interference assumption is imposed or not.

The following theorem derives the asymptotic behavior of the estimator:
\begin{theorem}[CLT for AAIPW estimator]\label{thm:clt} 

Suppose Assumption \ref{ass:consistency}-\ref{ass:positivity} holds. Under Assumption \ref{ass:estimator_consistent}, \ref{ass:network_weak_dependence} and regularity assumptions in Appendix \ref{app:proof-clt}, we have:
    \begin{align*}
          \sqrt{n}\left(\widehat \mu_{\mathrm{AIPW}}-\mu_n \right) \overset{d}{\longrightarrow} \mathcal{N}(0,\Lambda).
    \end{align*}
    where $\Lambda = \frac{1}{n}\Var(\sum_{i=1}^n W_i)$ and 
    \begin{align}
        W_i &= \frac{1\{ \mathbf{A}_{\cN_i} = \mathbf{a}_{\cN_i}\} }{\pi_{\mathbf{a}_{\cN_i}}(X_i^a;\eta)} \left(Y_i-\beta_{\mathbf{a}_{\cN_i}}(X_i^y;\theta)\right) +\beta_{\mathbf{a}_{\cN_i}}(X_i^y;\theta) -  \mu_i \label{eqn:dr-if}\\
        &\hspace{6mm} - M_1 \cdot \mathrm{IF}_{ \eta}(\mathbf{A}_{\cN_i},X_i^a)  + M_2 \cdot \mathrm{IF}_{\theta}(Y_i, X_i^y)  ~,\label{eqn:error-if}
    \end{align}
    where
    \begin{align*}
        M_1 &:= \lim_{n\to\infty}\frac{1}{n}\sum_{i=1}^n\mathrm{\mE}\left[ \frac{1\{ \mathbf{A}_{\cN_i} = \mathbf{a}_{\cN_i}\} }{\left(\pi_{\mathbf{a}_{\cN_i}}(X_i^a;\eta)\right)^2} \frac{\partial \pi_{\mathbf{a}_{\cN_i}}(X_i^a;\eta)}{\partial \eta} \left(Y_i-\beta_{\mathbf{a}_{\cN_i}}(X_i^y;\theta)\right)\right]~,\\
        M_2 &:= \lim_{n\to\infty}\frac{1}{n}\sum_{i=1}^n\mathrm{\mE}\left[ \frac{1}{n} \sum_{i=1}^n \left(1 - \frac{1\{ \mathbf{A}_{\cN_i} = \mathbf{a}_{\cN_i}\} }{\pi_{\mathbf{a}_{\cN_i}}(X_i^a;\eta)} \right) \frac{\partial \beta_{\mathbf{a}_{\cN_i}}(X_i^y;\theta)}{\partial \theta} \right]~,
    \end{align*}
    where $\text{IF}_{\theta}(Y_i, X_i^y)$ and $\text{IF}_{\eta}(\mathbf{A}_{\cN_i},X_i^a)$ represent the influence function for the outcome and propensity model estimation, respectively. More specifically, if we use method of moments for both models, then
    \begin{align*}
    \text{IF}_{\theta}(Y_i, X_i^y) &= -M_3^{-1} g^y(Y_i, X_i^y;\theta) ~,\\
    \text{IF}_{\eta}(\mathbf{A}_{\cN_i},X_i^a) &= -M_4^{-1} g^a(\mathbf{A}_{\cN_i}, X_i^a;\eta)~,
    \end{align*}
    where $\mE[g^y(Y_i, X_i^y;\theta)]=0$ and $\mE[g^a(\mathbf{A}_{\cN_i}, X_i^a; \eta)]=0$ represent the moment conditions, and
    \begin{align*}
        &M_3:=\lim_{n\to\infty}\frac{1}{n}\sum_{i=1}^n\mE\left[\frac{\partial g^y(Y_i,X_i^y;\theta)}{\partial\theta}\right]\\
         &M_4: = \lim_{n\to\infty}\frac{1}{n}\sum_{i=1}^n\mE\left[\frac{\partial g^a(Y_i,X_i^a;\eta)}{\partial\eta}\right]~.
    \end{align*}
\end{theorem}
\begin{remark}
    Note that \eqref{eqn:dr-if}) does not correspond to our parameter of interest unless we assume either $\beta_{\mathbf{a}_{\cN_i}}(X_i^y;\theta) = \beta_{\mathbf{a}_{\cN_i}}(X_i^y)$ or $\pi_{\mathbf{a}_{\cN_i}}(X_i^a;\eta) = \pi_{\mathbf{a}_{\cN_i}}(X_i^a)$. Moreover,
    \eqref{eqn:error-if} characterizes the contribution of model misspecification to the influence function in \eqref{eqn:dr-if}. Specifically,
    \begin{align*}
        M_1 &= 0 \quad \text{if Assumption~\ref{ass:param}(i) holds and $\theta_0 = \theta$}~, \\
        M_2 &= 0 \quad \text{if Assumption~\ref{ass:param}(ii) holds and $\eta_0 = \eta$}~.
    \end{align*}
    That is, the asymptotic distribution varies depending on whether one, both, or neither of the models are correctly specified.
\end{remark}

\begin{remark}\label{remark:hac-inference}
    The asymptotic result in Theorem~\ref{thm:clt}, together with the framework developed in \cite{kojevnikov2021limit}, yields an immediate and practical approach to statistical inference. In particular, we adopt the Network HAC variance estimator proposed in Section 4 of \cite{kojevnikov2021limit}, which takes the form:
    \begin{equation*}
        \widehat{\Lambda}_n = \sum_{s \geq 0} \omega_n(s) \widehat{\Omega}_n(s)~,
    \end{equation*}
    where
    \begin{equation*}
        \widehat{\Omega}_n(s) = \frac{1}{n} \sum_{1 \leq i \leq n} \sum_{j \in \cN_n^\partial(i; s)} \left(\widehat W_i - \widehat\mu_{\mathrm{AIPW}}\right) \left(\widehat W_j - \widehat\mu_{\mathrm{AIPW}}\right)~,
    \end{equation*}
    and $\omega_n(s)$ is a user-specified kernel function. In our simulation study, we employ the commonly used Bartlett kernel.
\end{remark}



\subsection{Identification with Exposure Mapping} \label{subsec:exposure}

In this section, we extend the identification framework to accommodate a common modeling approach in the presence of network interference: the use of exposure mappings. Exposure mappings impose a structured form of interference by assuming that a unit's potential outcome depends on a lower-dimensional summary of the treatment assignment vector, rather than the full vector itself. This assumption is particularly appealing in practice, as it reduces the complexity of the treatment space and facilitates estimation when the local neighborhood of a unit is large.

The key motivation behind this extension is the high dimensionality of the treatment assignment vector $\mathbf{A}_{\cN_i}$, particularly in dense networks where each unit may have many neighbors. In such cases, the number of possible configurations of $\mathbf{A}_{\cN_i}$ grows exponentially with the neighborhood size, rendering nonparametric identification and estimation infeasible. Exposure mappings provide a principled way to project $\mathbf{A}_{\cN_i}$ onto a lower-dimensional space while preserving the relevant causal structure.

Formally, let $\mathscr{A}^*(n) = \cup_{m=1}^n \mathscr{A}(m)$ denote the set of all binary vectors of length up to $n$, representing all possible local treatment assignments. We define a general exposure mapping as a function $T: \mathscr{A}^*(n) \to \mathcal{T}$, where $\mathcal{T} \subset \mathbb{R}^{d}$ is a finite set. Given a fixed neighborhood size $K$, we define $T_i = T(\mathbf{A}_{\cN_i})$ to be the exposure level of unit $i$ implied by its local assignment vector.

Several examples of exposure mappings used in the literature illustrate this structure. One canonical example is
$
T_i = (A_i, \textstyle\sum_{j \in \cN_n^\partial(i; 1)} A_j),
$
which records unit $i$'s own treatment status along with the number of treated neighbors. These mappings reduce the dimensionality of the local assignment space from $2^{|\cN_n(i;1)|}$ to $|\cN_n^\partial(i; 1)|$, a more manageable number of exposure levels, enabling more efficient estimation of exposure-specific average potential outcomes.

The following theorem presents a new moment function under the exposure mapping framework and establishes an identification result using only the ignorability assumption.
\begin{theorem}\label{thm:exposure-1}
    Under Assumption \ref{ass:consistency}-\ref{ass:positivity}, for any $K\geq 1$, we have:
    \begin{align}\label{eqn:exposure-1}
    \begin{aligned}
        &\mE\left[  \frac{1\left\{ T_i = t \right\} \left(Y_i-\mE[Y_i \mid \mathbf{Y}_{-i}, \mathbf{A}_{\cN_i}=\mathbf{a}_{\cN_i}, \mathbf{A}_{-\cN_i}, \mathbf{L}]\right)}{\mP\left( T_i = t\mid \mathbf{L}, \mathbf{A}_{-\cN_i} \right)} + \mE[Y_i \mid \mathbf{Y}_{-i}, \mathbf{A}_{\cN_i}=\mathbf{a}_{\cN_i}, \mathbf{A}_{-\cN_i}, \mathbf{L}]\right]\\
        &= \mE\left[\sum_{\mathbf{a}_{\cN_i}} \mP(\mathbf{A}_{\cN_i}=\mathbf{a}_{\cN_i}\mid T_i = t,\mathbf{L},\mathbf{A}_{-\cN_i})  Y_i(\mathbf{a}_{\cN_i}, \mathbf{A}_{-\cN_i}) \right] ~.
    \end{aligned}
    \end{align}
\end{theorem}
Note that the network propensity score in Theorem~\ref{thm:exposure-1} is a marginalized version of the original propensity score in Theorem~\ref{thm:identification1} over all treatment allocations compatible with $T_i = t$. That is,
\[
\mP\left( T_i = t \mid \mathbf{L}, \mathbf{A}_{-\cN_i} \right) = \sum_{\mathbf{a}_{\cN_i}} \mP\left( \mathbf{A}_{\cN_i} = \mathbf{a}_{\cN_i} \mid \mathbf{L}, \mathbf{A}_{-\cN_i} \right) 1\{T(\mathbf{a}_{\cN_i})=t\}~.
\]

To sharpen identification, we introduce additional structure through Assumption~\ref{ass:exposure}. Assumption~\ref{ass:exposure}(i) imposes a restricted interference condition, requiring that potential outcomes depend only on the exposure mapping. Assumption~\ref{ass:exposure}(ii) imposes a corresponding restriction on outcomes, stating that the conditional expectation depends only on the exposure level rather than the full assignment vector. 

\begin{assumption}[Model with Exposure Mapping]
\label{ass:exposure} 
\quad
\begin{enumerate}
    \item[(i)] For all $i \in [n]$ and $\mathbf{a}, \mathbf{a}^\prime \in \mathscr{A}(n)$,
    \begin{equation*}
        Y_i(\mathbf{a}) = Y_i(\mathbf{a}^\prime) \quad \text{w.p.1 if } T(\mathbf{a}_{\cN_i}) = T(\mathbf{a}^\prime_{\cN_i})~.
    \end{equation*}
    \item[(ii)] The observed outcome satisfies
    \begin{equation*}
        \mE[Y_i \mid \mathbf{Y}_{-i}, \mathbf{A}=\mathbf{a}, \mathbf{L}] = \mE[Y_i \mid \mathbf{Y}_{-i}, T_i=t, \mathbf{L}]~,
    \end{equation*}
    where $T_i = T(\mathbf{A}_{\cN_i})$ and $t = T(\mathbf{a}_{\cN_i})$.
\end{enumerate}
\end{assumption}

The next result shows that, under ignorability, either condition in Assumption~\ref{ass:exposure} is sufficient to identify the marginal potential outcome $\mE[Y_i(\mathbf{a})]$. This theorem parallels our earlier results and illustrates a double robustness property in the presence of exposure mappings.

\begin{theorem}\label{thm:exposure-2}
    Suppose Assumption \ref{ass:consistency}-\ref{ass:positivity} holds. Then, under either Assumption \ref{ass:exposure}(i) or Assumption \ref{ass:exposure}(ii), for any $K\geq 1$, we have
    \begin{align*}
        &\mE\left[  \frac{1\left\{ T_i = t \right\} \left(Y_i-\mE[Y_i \mid \mathbf{Y}_{-i}, \mathbf{A}_{\cN_i}=\mathbf{a}_{\cN_i}, \mathbf{A}_{-\cN_i}, \mathbf{L}]\right)}{\mP\left( T_i = t\mid \mathbf{L}, \mathbf{A}_{-\cN_i} \right)} + \mE[Y_i \mid \mathbf{Y}_{-i}, \mathbf{A}_{\cN_i}=\mathbf{a}_{\cN_i}, \mathbf{A}_{-\cN_i}, \mathbf{L}]\right]\\
        &=  \mE[Y_i(\mathbf{a})] ~.
    \end{align*}
\end{theorem}

Finally, the estimation strategies developed in Sections~\ref{subsec:main2} and~\ref{subsec:main3clt} can be directly applied in this setting, using the exposure mapping to construct estimators with desirable theoretical guarantees.

\section{Simulation Study}\label{sec:simulation}
We conducted an extensive simulation study to evaluate the performance of the proposed methods across networks of varying density and size. Specifically, we examined the double robustness property of the estimator, assessed its empirical performance relative to the outcome-model-based Auto-G-Computation method of \cite{auto-g}, and evaluated the validity of statistical inference based on the network HAC variance estimator. The parameters of interest include the four estimands listed in Table~\ref{table:estimands}.\footnote{In all simulation studies, we set $\alpha = 0.7$ and $\alpha^\prime = 0.3$.}

\begin{table}[ht!]
\centering
\setlength{\tabcolsep}{3.5pt}
\begin{tabular}{ll}
\toprule
Estimand & Formula \rule[2ex]{0pt}{-2ex} \\ \midrule
$\gamma(\alpha)$ & $n^{-1} \sum_{i=1}^n \sum_{\mathbf{a} \in \mathscr{A}(n)} \mE[Y_i(\mathbf{a})] \pi_i(\mathbf{a}; \alpha)$ \rule[-3.5ex]{0pt}{2ex} \\
$\DE(\alpha)$ & $n^{-1} \sum_{i=1}^n \sum_{\mathbf{a}_{-i} \in \mathscr{A}(n-1)} \left( \mE[Y_i(\mathbf{a}_{-i}, a_i = 1)] - \mE[Y_i(\mathbf{a}_{-i}, a_i = 0)] \right) \pi_i(\mathbf{a}_{-i}; \alpha)$ \rule[-3.5ex]{0pt}{2ex} \\
$\IE(\alpha)$ & $n^{-1} \sum_{i=1}^n \sum_{\mathbf{a}_{-i} \in \mathscr{A}(n-1)} \left( \mE[Y_i(\mathbf{a}_{-i}, a_i = 0)] - \mE[Y_i(\mathbf{a}_{-i} = \mathbf{0}, a_i = 0)] \right) \pi_i(\mathbf{a}_{-i}; \alpha)$ \rule[-3.5ex]{0pt}{2ex} \\
$\IE(\alpha, \alpha')$ & $n^{-1} \sum_{i=1}^n \sum_{\mathbf{a}_{-i} \in \mathscr{A}(n-1)}  \mE[Y_i(\mathbf{a}_{-i}, a_i = 0)] \pi_i(\mathbf{a}_{-i}; \alpha) - \mE[Y_i(\mathbf{a}_{-i}, a_i = 0)]  \pi_i(\mathbf{a}_{-i}; \alpha^\prime)$ \\
\bottomrule
\end{tabular}
\caption{Parameters of interest for simulation studies}
\label{table:estimands}
\end{table}

We generate networks of size 800 using a modified Barabási-Albert (BA) model with a cap on node degrees to prevent overly connected hubs. The construction begins with $m$ isolated nodes. Each new node is then sequentially connected to $m$ existing nodes, selected preferentially based on their degree (i.e., nodes with higher degrees are more likely to be chosen), but only among those whose degrees are below a specified \texttt{max\_degree} threshold. The parameter $m$ determines both the number of initial nodes and the number of edges each new node introduces. Consequently, $m$ has a direct effect on the \textit{mean degree} of the resulting network: since each new node contributes exactly $m$ edges and each edge increases the degree of two nodes, the expected average degree approaches approximately $2m$ as the network grows. Larger values of $m$ therefore yield denser networks with higher average degree, while the \texttt{max\_degree} constraint prevents the emergence of extreme hubs. This construction is well-suited for evaluating the empirical performance of our methods, which rely on weak dependence assumptions discussed in Section~\ref{subsec:main3clt}.

\begin{table}[htbp]
\small
\setlength{\tabcolsep}{5pt}
\renewcommand{\arraystretch}{0.75}
\centering
\begin{tabular}{
  l
  c 
  >{\centering\arraybackslash}p{1.3cm}
  >{\centering\arraybackslash}p{1.3cm}
  >{\centering\arraybackslash}p{1.3cm}
  c
  >{\centering\arraybackslash}p{1.3cm}
  >{\centering\arraybackslash}p{1.3cm}
  >{\centering\arraybackslash}p{1.3cm}
}
\toprule
& & \multicolumn{3}{c}{\footnotesize\textbf{Misspecified Outcome Model}} & & \multicolumn{3}{c}{\footnotesize\textbf{Misspecified Propensity Model}} \\
\cmidrule(lr){3-5} \cmidrule(lr){7-9}
($m$, \texttt{max\_degree}) & & (1,2) & (2,5) & (3,10) & & (1,2) & (2,5) & (3,10) \\
\midrule
\multicolumn{9}{l}{$\gamma(\alpha)$} \\
True              & & 0.702 & 0.703 & 0.649 & & 0.702 & 0.703 & 0.649 \\
AAIPW Bias         & & 0.011 & -0.001 & -0.005 & & 0.002 & 0.001 & -0.000 \\
AAIPW RMSE         & & 0.052 & 0.038 & 0.027 & & 0.024 & 0.017 & 0.012 \\
Auto-G Bias       & & 0.087 & 0.057 & 0.030 & & 0.003 & 0.001 & -0.000 \\
Auto-G RMSE       & & 0.089 & 0.059 & 0.032 & & 0.028 & 0.021 & 0.014 \\
\midrule
\multicolumn{9}{l}{$\DE(\alpha)$} \\
True              & & -0.230 & -0.311 & -0.402 & & -0.230 & -0.311 & -0.402 \\
AAIPW Bias         & & 0.016 & -0.005 & -0.019 & & 0.002 & 0.002 & 0.001 \\
AAIPW RMSE         & & 0.078 & 0.058 & 0.048 & & 0.036 & 0.027 & 0.020 \\
Auto-G Bias       & & 0.187 & 0.121 & 0.085 & & 0.002 & 0.000 & 0.003 \\
Auto-G RMSE       & & 0.189 & 0.123 & 0.088 & & 0.036 & 0.027 & 0.021 \\
\midrule
\multicolumn{9}{l}{$\IE(\alpha, \alpha')$} \\
True              & & 0.178 & 0.202 & 0.207 & & 0.178 & 0.202 & 0.207 \\
AAIPW Bias         & & 0.009 & 0.020 & 0.035 & & 0.004 & 0.004 & 0.004 \\
AAIPW RMSE         & & 0.038 & 0.059 & 0.099 & & 0.018 & 0.027 & 0.053 \\
Auto-G Bias       & & 0.074 & 0.203 & 0.189 & & 0.070 & 0.164 & 0.139 \\
Auto-G RMSE       & & 0.078 & 0.205 & 0.190 & & 0.076 & 0.166 & 0.141 \\
\midrule
\multicolumn{9}{l}{$\IE(\alpha)$} \\
True              & & 0.354 & 0.518 & 0.618 & & 0.354 & 0.518 & 0.618 \\
AAIPW Bias         & & 0.013 & 0.070 & 0.149 & & 0.001 & 0.006 & 0.003 \\
AAIPW RMSE         & & 0.081 & 0.210 & 0.274 & & 0.037 & 0.096 & 0.159 \\
Auto-G Bias       & & 0.093 & 0.164 & 0.178 & & -0.000 & 0.003 & -0.002 \\
Auto-G RMSE       & & 0.101 & 0.169 & 0.182 & & 0.041 & 0.060 & 0.090 \\
\bottomrule
\end{tabular}
\caption{Estimation accuracy of AAIPW and Auto-G across different causal parameters, network densities, and model misspecification scenarios. Results are averaged over 8,000 Monte Carlo samples generated via Gibbs sampling.}
\label{table:aipw_autog_all}
\end{table}

\begin{table}[htbp]
\small
\renewcommand{\arraystretch}{0.75}
\setlength{\tabcolsep}{5pt}
\centering
\label{tab:coverage_panels_reduced}
\begin{tabular}{lcccccccccc}
\toprule
($m$, \texttt{max\_degree})& (1,2) & (1,4) & (1,8) & (1,10) & (2,5) & (2,6) & (2,8) & (2,10) & (3,10) & (4,10) \\
\midrule
$\gamma(\alpha)$       & 0.908 & 0.966 & 0.972 & 0.975 & 0.958 & 0.970 & 0.973 & 0.982 & 0.991 & 0.994 \\
$\DE(\alpha)$        & 0.921 & 0.959 & 0.960 & 0.968 & 0.950 & 0.966 & 0.964 & 0.973 & 0.982 & 0.994 \\
$\IE(\alpha, \alpha')$   & 0.942 & 0.941 & 0.957 & 0.948 & 0.939 & 0.930 & 0.943 & 0.935 & 0.944 & 0.940 \\
$\IE(\alpha)$    & 0.948 & 0.931 & 0.942 & 0.932 & 0.927 & 0.923 & 0.920 & 0.904 & 0.807 & 0.464 \\
\bottomrule
\end{tabular}
\caption{Coverage probabilities of 95$\%$ confidence intervals for causal effect estimates across different network densities. Results are averaged over 8,000 Monte Carlo samples generated via Gibbs sampling.}
\label{table:coverage}
\end{table}

To investigate the double robustness property, we consider two scenarios: (i) only the propensity model is correctly specified; and (ii) only the outcome model is correctly specified. Results for the scenario in which both models are correctly specified are available in Appendix~\ref{app:extra_sim}, but are omitted here for brevity. In each misspecified model scenario, we adopt an adversarial setup in which the input variables to the misspecified model are entirely uninformative random noise. For the correctly specified model, we use input variables implied by the parametric Chain Graph model that governs the data generating process (DGP). Full details of the DGP are provided in Appendix~\ref{app:simulation-details}. In our setting, the variables $Y_i$, $A_i$, and $L_i$ are all binary, so logistic regression yields a correctly specified model for both the outcome and the propensity score. Estimation of the propensity score proceeds in two steps: we first fit a logistic regression to estimate individual-level propensity scores, and then use a closed-form expression derived from the Chain Graph model to compute the joint propensity of an individual’s treatment and their neighbors’ treatments. See Appendix~\ref{subsec:sim-estimator} for further details.

To benchmark the performance of our method, we compare it against the Auto-G-Computation estimator from \cite{auto-g} under all three model specification scenarios. The results are summarized in Table~\ref{table:aipw_autog_all}. In nearly all cases, the AAIPW estimator outperforms Auto-G in terms of both estimation bias and RMSE. In particular, when the outcome model is misspecified, Auto-G exhibits substantial bias, whereas the bias of the AAIPW estimator remains small. Furthermore, by examining the performance of AAIPW under both model misspecification scenarios, we demonstrate that the estimator is consistent under the union model, thereby validating its double robustness. One exception arises in the estimation of $\IE(\alpha)$ when the network is dense, where AAIPW performs relatively poorly. This is due to the high variance of the inverse propensity score weights when estimating $\mE[Y_i(\mathbf{a}_{-i} = \mathbf{0}, a_i = 0)]$. In dense networks, the probability that all units remain untreated can be extremely small, resulting in large inverse weights and, consequently, high variance in the estimator.

To evaluate the statistical inference based on Network HAC estimation (Remark~\ref{remark:hac-inference}), we report coverage rates under varying network densities in Table~\ref{table:coverage}. The coverage rates are generally close to 95\%, indicating that the proposed inference procedure performs well in most settings. However, there are a few challenging cases where coverage falls substantially below 95\%, notably for the parameter $\IE(\alpha)$ when $m = 4$ and $\text{max\_degree} = 10$. This deterioration arises from the high density of the network combined with the increased variance of the AAIPW estimator, as $\IE(\alpha)$ involves estimating $\mE[Y_i(\mathbf{0})]$.

\section{Data Application}\label{sec:empirical}
In this section, we illustrate the estimation and inference methods introduced in Section~\ref{sec:main-result} using data collected in the Networks, Norms, and HIV/STI Risk Among Youth (NNAHRAY) study. The NNAHRAY study was conducted in an impoverished neighborhood of Brooklyn, New York, from 2002 to 2005, in the context of an HIV epidemic and widespread drug use \citep{Friedman17}. Through in-person interviews, information was collected regarding the respondents’ demographic characteristics, incarceration history, sexual partnerships and histories, and past drug use. The study population we consider includes all interviewed persons for a total sample size of $n=465$ persons. Following \cite{auto-g}, we define network tie (i.e. edge) as a sexual and/or injection drug use partnership in the past three months if at least one of the partners reported the relationship. The network structure is given in Figure \ref{fig:network}.

\begin{figure}[h]
    \centering
    \includegraphics[width=0.9\linewidth]{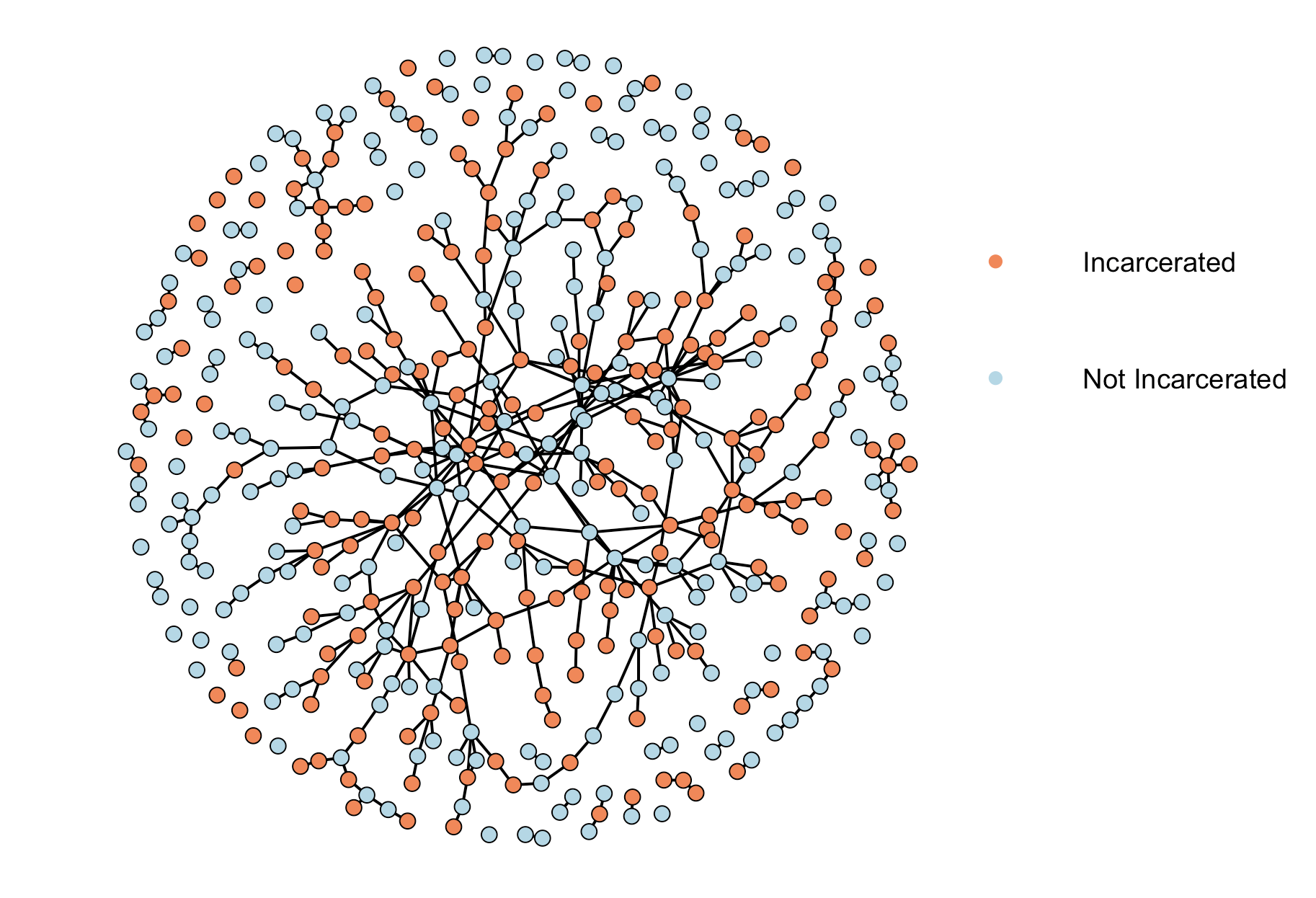}
    \caption{Network graph from the NNAHRAY data ($n=465$)}
    \label{fig:network}
\end{figure}

Our empirical study analyzes the impact of past incarceration on several socioeconomic outcomes of interviewed residents in the neighborhood, including employment, household income, and homelessness. We define past incarceration as any history of jail time reported by respondents. To account for confounding, we adjust for Latino/a ethnicity, age, education, and past illicit drug use. We employ the same models and estimation procedures described in the simulation section. The causal parameters of interest include the direct effect and spillover effects, $\DE(\alpha)$, $\IE(\alpha)$, and $\IE(\alpha, \alpha^\prime)$, as defined in Table~\ref{table:estimands}, with $\alpha = 0.5$ and $\alpha^\prime = 0.2$.

\begin{table}[ht]
\small
\renewcommand{\arraystretch}{0.75}
\label{table:empircal}
\centering
\begin{tabular}{lcccc}
\toprule
\textbf{Effect} & \multicolumn{2}{c}{\textbf{AAIPW}} & \multicolumn{2}{c}{\textbf{Auto-G}} \\
\cmidrule(lr){2-3} \cmidrule(lr){4-5} 
& Estimate & 95\% CI  & Estimate & 95\% CI \\
\midrule
\textbf{Employment} & & & &  \\
$\gamma(\alpha)$      & 0.188 & (0.134, 0.243) & 0.187 & (0.176, 0.225) \\
 $\DE(\alpha)$     & -0.159 & (-0.262, -0.055) & -0.100 & (-0.137, -0.066) \\
$\IE(\alpha, \alpha')$    & -0.071 & (-0.139, -0.002) & -0.030 & (-0.058, -0.008) \\
$\IE(\alpha)$   & -0.164 & (-0.301, -0.027) & -0.053 & (-0.101, -0.014) \\
\midrule
\textbf{Household Income} & & & & \\
$\gamma(\alpha)$      & 0.151 & (0.094, 0.208) & 0.134 & (0.113, 0.165) \\
 $\DE(\alpha)$      & 0.048 & (-0.042, 0.139) & 0.014 & (-0.007, 0.036) \\
$\IE(\alpha, \alpha')$    & 0.041 & (-0.008, 0.089) & 0.002 & (-0.008, 0.013) \\
$\IE(\alpha)$    & 0.043 & (-0.037, 0.123) & 0.003 & (-0.013, 0.021) \\
\midrule
\textbf{Homelessness} & & & & \\
$\gamma(\alpha)$     & 0.184 & (0.119, 0.249) & 0.125 & (0.107, 0.173) \\
 $\DE(\alpha)$      & 0.167 & (0.016, 0.318) & 0.058 & (0.036, 0.095) \\
$\IE(\alpha, \alpha')$   & -0.016 & (-0.061, 0.029) & 0.003 & (-0.005, 0.014) \\
$\IE(\alpha)$    & -0.076 & (-0.161, 0.010) & 0.006 & (-0.007, 0.023) \\
\bottomrule
\end{tabular}
\caption{Estimates of the causal parameters in Table~\ref{table:estimands} for employment, household income, and homelessness using the AAIPW and Auto-G methods. Inference for AAIPW is conducted via the Network HAC procedure described in Remark~\ref{remark:hac-inference}, whereas Auto-G inference is based on 1,000 bootstrap replications.}
\label{table:effects}
\end{table}

Table~\ref{table:effects} presents the empirical estimates for each outcome. We find strong evidence of negative direct and spillover effects of incarceration on employment, indicating that both a respondent's own incarceration and the incarceration of others in their network reduce their likelihood of employment. For household income, we do not find significant direct or spillover effects, suggesting that incarceration may not have a detectable impact on income levels in this setting, as the targeted study population was overall economically disadvantaged, irrespective of their history of incarceration. For homelessness, we observe a significant positive direct effect, implying that individuals with a history of incarceration are more likely to experience homelessness. However, we do not detect any spillover effects of incarceration on homelessness among connected individuals. The AAIPW and Auto-G estimators lead to qualitatively similar conclusions across all outcomes and estimands of interest. These patterns hold consistently across both estimation approaches, with the AAIPW and Auto-G estimators yielding qualitatively similar conclusions across all outcomes and estimands of interest.

\section{Conclusion}\label{sec:conclusion}
We have proposed a new framework for causal inference in network data by extending the Augmented Inverse Propensity Weighted (AIPW) estimator to accommodate both the dependence and interference induced by network structures. In i.i.d. data settings, AIPW estimators are known to achieve the semi-parametric efficiency bound. Extending this notion of semi-parametric efficiency to the network context—particularly when using nonparametric estimation techniques for the nuisance models—represents a promising avenue for strengthening the theoretical guarantees of our approach.

Additionally, while our focus here has been on a single large network, many empirical datasets contain such a network along with a large number of isolated units, as exemplified in the network structure shown in Figure~\ref{fig:network}. Developing strategies to leverage these isolated units—who are not subject to interference and may provide valuable information—could potentially improve both the robustness and efficiency of causal effect estimation in practice.  

Finally, throughout, we have made the key ignorability Assumption~\ref{ass:ignorability}, an assumption which may not hold in nonexperimental settings due to unmeasured confounding. It would be interesting to consider alternative identification strategies that do not require ignorability, such as either network versions of instrumental variable methods (\cite{BRAMOULLE2009, Goldsmith-Pinkham2013,Manski1993, omalley2014}; or difference-in-differences type approachs (\cite{egami2020}); or recent developments in proximal causal inference for network confounding or latent homophily (\cite{Egami2023}). 
Exploring these ideas in greater detail is an important direction for our future work.

\section*{Acknowledgement}
We are grateful to Dr. Samuel R. Friedman at National Development and Research Institutes, Inc. for access to the Networks, Norms, and HIV/STI Risk Among Youth study data.

The authors would like to acknowledge the National Institutes of Health (NIH) for their generous funding and support.

\clearpage
\newpage

\phantomsection\label{supplementary-material}
\bigskip

\begin{center}

{\large\bf SUPPLEMENTARY MATERIAL}

\end{center}

\appendix


\section{A Graphical Statistical Model for Network Data}\label{app:chain-graph}
\subsection{A Chain Graph Representation of Network Data}\label{subsec:chain-graph-intro}
This section introduces a graphical framework for modeling network data using chain graphs, as proposed by \cite{auto-g,shpitser2024}. A chain graph is a mixed graph containing both directed and undirected edges, with the restriction that the undirected edges cannot be oriented in a way that forms a directed cycle \citep{lauritzen1996graphical, lauritzen2002chain}. Chain graphs generalize both directed acyclic graphs (DAGs) and undirected graphs, and are particularly suited for settings involving both symmetric associations (e.g., network ties) and asymmetric causal mechanisms (e.g., treatment-outcome relationships).

We now describe how to construct the chain graph model \(\mathcal{G}_{\mathcal{E}_n}\) based on a given network structure. Recall that  \({\mathcal{E}_n}\) denote the set of undirected edges in the network, where \((i,j) \in {\mathcal{E}_n}\) if units \(i\) and \(j\) are directly connected, and for each unit \(i\), we observe covariates, treatment, and outcome, denoted by \(Y_i\), \(A_i\) and \(L_i\), respectively. In \(\mathcal{G}_{\mathcal{E}_n}\), we add \textit{undirected edges} between \(L_i\) and \(L_j\), \(A_i\) and \(A_j\), and \(Y_i\) and \(Y_j\) for each \((i,j) \in {\mathcal{E}_n}\), reflecting symmetric associations between neighboring units. We also include \textit{directed edges} to capture plausible causal pathways: specifically, for each \((i,j) \in {\mathcal{E}_n}\), we add
\[
L_i \to A_i, \quad L_i \to A_j, \quad L_i \to Y_i, \quad L_i \to Y_j, \quad A_i \to Y_i, \quad A_i \to Y_j~.
\]
This construction allows individual covariates to influence both own and neighbor treatments and outcomes, and allows for interference through treatment and covariate pathways.

The global Markov property of a chain graph $\cD$ asserts that whenever two sets $B_1$ and $B_2$ are $c-$separable given a set $S$. We have 
\begin{align*}
    B_1\independent B_2| S.
\end{align*}
The concept of $c$-separation for chain graphs coincides with standard $g$-separation in undirected graphs and with $d$-separation in DAGs \citep{studeny2013bayesian}. Two sets $B_1$ and $B_2$ are said to be $c$-separable given $S$ if every walk between $B_1$ and $B_2$ is blocked by $S$. See \cite{studeny2013bayesian} for details of this property, including the formal definitions of walks and blocking in chain graph models.

Applying the global Markov property to the chain graph $\mathcal{G}_{\mathcal{E}_n}$ constructed above yields the following intuitive conditional independence relationships:
\begin{align}\label{eq:markov-y}
Y_i &\independent \{Y_k, A_k, L_k\} \mid A_i, L_i, \{Y_j, A_j, L_j : j \in \mathcal{N}_n^\partial(i;1)\}, \quad \text{for all } k \notin \mathcal{N}_n(i;1)~,
\end{align}
and 
\begin{align}\label{eq:markov-a}
    \mathbf{A}_{\cN_n(i;\ell)} &\independent \{A_k, L_k\} \mid \left(\mathbf{L}_{\cN_n(i;\ell)}, \{A_j, L_j : j \in \cN_n^\partial(i;\ell+1)\}\right)
\end{align}
for all $i$ and $k\notin \cN_n(i;\ell+1)$ and $\ell \geq 0$.

\eqref{eq:markov-y} states that the outcome \(Y_i\) is conditionally independent of all variables associated with non-neighbors, given the covariates and treatments of \(i\) and its neighbors. Similarly, \eqref{eq:markov-a} implies that \(A_i\) is conditionally independent of distant covariates beyond distance $\ell +1$, once the corresponding information from all units within distance $\ell+1$ is conditioned on. These ``screening-off'' properties define the Markov blankets of \(Y_i\) and \(A_i\) in \(\mathcal{G}_{\mathcal{E}_n}\) and yield substantial dimension reduction in modeling the conditional outcome. In particular, the conditional density \(f(Y_i \mid \mathbf{Y}_{-i}, \mathbf{A}, \mathbf{L})\) depends only on data from local units, i.e. network ties, making estimation more tractable.

\subsection{Parametric Specifications of Auto-models}\label{subsec:chain-graph-parametric}
In this section, we describe the parametric auto-regression model, abbreviated as an auto-model, for the treatment vector \(\mathbf{A}\) conditional on the covariates \(\mathbf{L}\), based on the Chain Graph representation introduced above. This follows the approach of \cite{auto-g}, who focuses on modeling the outcome distribution \(f(\mathbf{Y} \mid \mathbf{A}, \mathbf{L})\); since the two specifications are structurally similar, we focus here on the propensity score model.

Under the chain graph model and assuming positivity, the conditional distribution of \(\mathbf{A}\) given \(\mathbf{L}\) admits the following factorization:
\begin{equation*}
    f(\mathbf{a} \mid \mathbf{l}) = \left(\frac{1}{\kappa(\mathbf{l})} \right)\exp\left\{U(\mathbf{a}; \mathbf{l})\right\}~,
\end{equation*}
where \(\kappa(\mathbf{l}) = \sum_{\mathbf{a}} \exp\left\{U(\mathbf{a}; \mathbf{l})\right\}\) is the normalizing constant, and \(U(\mathbf{a}; \mathbf{l})\) is a conditional energy function that aggregates contributions from a small set of local variables. Following the framework of \cite{Besag1974}, we assume the energy function takes the form:
\begin{equation*}
    U(\mathbf{a}; \mathbf{l}) = \sum_{i \in \mathcal{G}_{{\mathcal{E}_n}}} a_i G_i(a_i; \mathbf{l}) + \sum_{\{i,j\} \in {\mathcal{E}_n}} a_i a_{j} \eta_{ij}(\mathbf{l})~,
\end{equation*}
where \(G_i(a_i; \mathbf{l})\) and \(\eta_{ij}(\mathbf{l})\) are local functions and interaction coefficients, respectively.

We now specify parametric forms for these functions:
\begin{align*}
    G_i(a_i; \mathbf{l}) &= \widetilde{G}(\mathbf{l}_{\cN_n(i,1)}) = \iota_0 + \iota_1 l_i + \iota_2 \sum_{j \in \cN^\partial_n(i,1)} w_{ij}^l l_j~,\\
    \eta_{ij}(\mathbf{l}) &= \widetilde{\eta}_{ij} = w_{ij}^y \eta~,
\end{align*}
where \(w_{ij}^l\) and \(w_{ij}^y\) are user-specified weights encoding the importance of neighbor covariates and treatment interactions, respectively. 

Given this specification, the network propensity score for unit \(i\) is given by:
\begin{align}\label{eqn:chain-graph-propensity}
    &\mP(\mathbf{A}_{\cN_i} = \mathbf{a}_{\cN_i} \mid \mathbf{A}_{-\cN_i}, \mathbf{L}) = \frac{\mP(\mathbf{A}_{\cN_i} = \mathbf{a}_{\cN_i}, \mathbf{A}_{-\cN_i} \mid \mathbf{L})}{\sum_{\mathbf{a}^*} \mP(\mathbf{A}_{\cN_i} = \mathbf{a}^*, \mathbf{A}_{-\cN_i} \mid \mathbf{L})} \nonumber\\ 
    &= \frac{\exp\left\{\sum_{k \in \cN_i} a_k \widetilde{G}(\mathbf{L}_{\cN_n(k,1)}) + \sum_{k \in \cN_i} \sum_{j \in \cN_n(k,1)\cap \cN_i} a_k a_j \widetilde{\eta}_{kj}+\sum_{k \in \cN_i} \sum_{j \in \cN_n(k,1)\setminus \cN_i} a_k A_j \widetilde{\eta}_{kj}\right\}}{\sum_{\mathbf{a}^*} \exp\left\{\sum_{k \in \cN_i} a_k^* \widetilde{G}(\mathbf{L}_{\cN_n(k,1)}) + \sum_{k \in \cN_i} \sum_{j \in \cN_n(k,1)\cap \cN_i} a_k^* a_j^* \widetilde{\eta}_{kj}+\sum_{k \in \cN_i} \sum_{j \in \cN_n(k,1)\setminus \cN_i} a_k^* A_j \widetilde{\eta}_{kj}\right\}}~,
\end{align}

Recall that $\cN_i = \cN_{n}(i,K)$, \eqref{eqn:chain-graph-propensity} implies that the network propensity score for unit \(i\) depends only on treatment assignments of units in \(\cN_n^\partial(i,K+1)\) and covariates in \(\cN_n(i,K+1)\). Moreover, the model takes the form of a multinomial logit, which facilitates practical estimation. This local dependence, induced by the chain graph structure, enables tractable and scalable inference in high-dimensional networks.

\section{Simulation Details and Extra Results}\label{app:simulation-details}
\subsection{Data-Generating Process}
This section outlines the data-generating process (DGP) used in our simulation study. The DGP consists of three components: covariate generation, treatment assignment, and outcome modeling. We specify the full conditional distributions for each component, and samples are drawn via Gibbs sampling. A total of 10,000 iterations are generated, with the first 2,000 discarded as burn-in, yielding 8,000 simulated datasets. The detailed conditional distributions for each component are provided below.

\textbf{Covariate Generation}

In all simulation settings, binary covariates are generated using a chain graph model, with the conditional distributions for Gibbs sampling defined as follows:
\begin{align*}
    &\mP\left(L_{1, i}=1 \mid {L}_{- 1, i},\left\{{L}_{j}: j \in \mathcal{N}^-_i\right\}\right)\\
    &\hspace{3mm}=\operatorname{expit}\left(\tau_1+\rho_{12} L_{2, i}+\rho_{13} L_{3, i}+\nu_{11} \sum_{j \in \mathcal{N}^-_i} L_{1, j}+\nu_{12} \sum_{j \in \mathcal{N}^-_i} L_{2, j}+\nu_{13} \sum_{j \in \mathcal{N}^-_i}L_{3, j}\right) \\
    &\mP\left(L_{2, i}=1 \mid {L}_{- 2, i},\left\{{L}_{j}: j \in \mathcal{N}^-_i\right\}\right)\\
    &\hspace{3mm}=\operatorname{expit}\left(\tau_2+\rho_{12} L_{1, i}+\rho_{23} L_{3, i}+\nu_{21} \sum_{j \in \mathcal{N}^-_i} L_{1, j}+\nu_{22} \sum_{j \in \mathcal{N}^-_i} L_{2, j}+\nu_{23} \sum_{j \in \mathcal{N}^-_i}L_{3, j}\right) \\
    &\mP\left(L_{3, i}=1 \mid {L}_{- 3, i},\left\{{L}_{j}: j \in \mathcal{N}^-_i\right\}\right)\\
    &\hspace{3mm}=\operatorname{expit}\left(\tau_3+\rho_{13} L_{1, i}+\rho_{23} L_{3, i}+\nu_{31} \sum_{j \in \mathcal{N}^-_i} L_{1, j}+\nu_{32} \sum_{j \in \mathcal{N}^-_i} L_{2, j}+\nu_{33} \sum_{j \in \mathcal{N}^-_i}L_{3, j}\right) ~,
\end{align*}
where $\mathcal{N}^-_i = \cN_i\setminus\{i\} =\cN_n(i,K)\setminus\{i\} $ for some given distance $K$. Specific parameters used in simulation studies are reported in Table~\ref{tab:parameters}.

\textbf{Treatment Assignment}

As described in Section \ref{subsec:chain-graph-parametric}, conditioning on all covariates $\mathbf{L}$, treatments are generated by $f(\mathbf{a}|\mathbf{L}) \propto\exp(U(\mathbf{a};\mathbf{L}))$ with some energy function $U(\mathbf{a};\mathbf{L})$. In our simulation studies, the energy function is set to be
\begin{align*}
    U(\mathbf{a};\mathbf{L}) = \sum_i a_i G_i(\mathbf{L}) +\eta_7\sum_{d(i,j)=1}a_ia_j~,
\end{align*}
where
\begin{align*}
    G_i(\mathbf{L}) = \eta_0 + (\eta_1,\eta_3,\eta_5)' L_i + (\eta_2,\eta_4,\eta_6)' \sum_{j \in \cN^\partial_n(i,1)} L_j~.
\end{align*}
With this energy function, the conditional distribution of $A_i$ is given by
\begin{align*}
    \mP(A_i=1|\mathbf{L},\mathbf{A_{-i}}) &= \frac{\mP(A_i=1,\mathbf{A_{-i}}|\mathbf{L})}{\mP(A_i=0,\mathbf{A_{-i}}|\mathbf{L})+\mP(A_i=1,\mathbf{A_{-i}}|\mathbf{L})}\\
    &= \frac{\exp\{G_i(\mathbf{L})+\eta_7\sum_{j\in\mathcal N_i} A_j\}}{1+\exp\{G_i(\mathbf{L})+\eta_7\sum_{j\in\mathcal N_i} A_j\}}\\
    &=\operatorname{expit}\left(G_i(\mathbf{L})+\eta_7\sum_{j\in\mathcal N_i} A_j\right)~.
\end{align*}
Consequently,  $\mathbf{A}$ can be generated via the conditional distribution
\begin{align}\label{eqn:dgp-a}
\begin{aligned}
    &\mP\left(A_i=1 \mid L_i,\left\{A_j, {L}_j: j \in \mathcal{N}^-_i\right\}\right)\\&=\operatorname{expit}\bigg(\eta_0+\eta_1 L_{1, i}+\eta_2\sum_{j \in \mathcal{N}^-_i} L_{1, j}+\eta_3 L_{2, i}+\eta_4 \sum_{j \in \mathcal{N}^-_i} L_{2, j}+\eta_5 L_{3, i}+\eta_6 \sum_{j \in \mathcal{N}^-_i} L_{3, j}+\eta_7 \sum_{j \in \mathcal{N}^-_i} A_j\bigg) ~.
\end{aligned}
\end{align}
Specific parameters used in the simulation studies are reported in Table~\ref{tab:parameters}.

\textbf{Outcome Model}

Similar to the treatment assignment, the outcome model is given by the conditional distribution
\begin{align}\label{eqn:dgp-y}
\begin{aligned}
     &\mP\left(Y_i=1 \mid A_i, L_i, \left\{Y_j, A_j, {L}_j: j \in \mathcal{N}^-_i\right\}\right)
     \\=&\operatorname{expit}\bigg(\theta_0 +\theta_1 A_i+\theta_2 \sum_{j \in \mathcal{N}^-_i} A_j+\theta_3 L_{1, i}+\theta_4 \sum_{j \in \mathcal{N}^-_i} L_{1, j}  \\
    &+\theta_5 L_{2, i}+\theta_6 \sum_{j \in \mathcal{N}^-_i} L_{2, j}+\theta_7 L_{3, i}+\theta_8 \sum_{j \in \mathcal{N}^-_i} L_{3, j}+\theta_9 \sum_{j \in \mathcal{N}^-_i} Y_j\bigg) ~.    
\end{aligned}
\end{align}

The specific parameters used in simulation studies are given in Table~\ref{tab:parameters}.
\begin{table}[htbp]
\small
\renewcommand{\arraystretch}{0.75}
\setlength{\tabcolsep}{5pt}
\centering
\begin{tabular}{llllllllll}
\toprule
\textnormal{Parameter} & Value & \textnormal{Parameter} & Value & \textnormal{Parameter} & Value & \textnormal{Parameter} & Value & \textnormal{Parameter} & Value \\
\midrule
\multicolumn{10}{l}{\textbf{Panel A: Covariate generation parameters}}\\
$\tau_1$    & -1.0     & $\tau_2$    & 0.5     & $\tau_3$    & -0.50     & $\rho_{12}$ & 0.1      & $\rho_{13}$ & 0.2   \\ 
$\rho_{21}$ & 0.1      & $\rho_{23}$ & 0.1      &$\rho_{31}$ & 0.2      & $\rho_{32}$ & 0.1          & $\nu_{11}$  & 0.1  \\    $\nu_{12}$  & 0        & $\nu_{13}$  & 0       &$\nu_{21}$  & 0.1       & $\nu_{22}$  & 0     & $\nu_{23}$  & 0      \\    $\nu_{31}$  & 0.1      & $\nu_{32}$  & 0        & $\nu_{33}$  & 0     & & & & \\    
\midrule
\multicolumn{10}{l}{\textbf{Panel B: Treatment assignment parameters}}\\
$\eta_0$  & -1       & $\eta_1$  & 2     & $\eta_2$  & 0.1        & $\eta_3$  & -2       & $\eta_4$  & 0.1     \\ 
 $\eta_6$  & 2        &$\eta_7$  & 0.1      & $\eta_8$  & 0.1      & & & & \\
 \midrule
\multicolumn{10}{l}{\textbf{Panel C: Outcome model parameters}}\\

$\theta_1$   & $-m$ &   $\theta_2$   &  $-2m$ & $\theta_3$ & 2   &$\theta_4$   & 2   & $\theta_5$   & 0.1      \\
$\theta_6$   & -1    & $\theta_7$   & 0.1  & $\theta_8$   & 2        & $\theta_9$   & 0.1      & $\theta_{10}$& 0              \\
\bottomrule
\end{tabular}
\caption{Simulation parameters for covariate construction, treatment allocation, and outcome modeling.}\label{tab:parameters}
\end{table}

\subsection{Estimation Procedure}\label{subsec:sim-estimator}
For each realization of the DGP, we construct the estimated causal effect following \eqref{eqn:AIPW}. Specifically, the estimator is given by 
$
\widehat \mu_{\mathrm{AIPW}} = \frac{1}{n} \sum_{i=1}^n \widehat W_i,
$
where
\begin{align*}
    \widehat W_i = \frac{1\{ \mathbf{A}_{\cN_i} = \mathbf{a}_{\cN_i} \}}{\pi_{\mathbf{a}_{\cN_i}}(X_i^a; \widehat\eta)} \left(Y_i - \beta_{\mathbf{a}_{\cN_i}}(X_i^y; \widehat\theta)\right) + \beta_{\mathbf{a}_{\cN_i}}(X_i^y; \widehat\theta)~.
\end{align*}
Thus, evaluating the estimator requires two components, $\beta_{\mathbf{a}_{\cN_i}}(X_i^y; \widehat\theta)$ and $\pi_{\mathbf{a}_{\cN_i}}(X_i^a; \widehat\eta)$. To calculate them, we first obtain $\widehat \eta$ and $\widehat\theta$ via logistic regressions corresponding to \eqref{eqn:dgp-a} and \eqref{eqn:dgp-y}, respectively. With $\widehat\theta$, $\beta_{\mathbf{a}_{\cN_i}}(X_i^y; \widehat\theta)$ can be calculated directly. In contrast, evaluating the propensity score requires additional steps. Although \eqref{eqn:dgp-a} provides the full conditional probability, our propensity score considers the treatments within a neighborhood of distance $K$. Note that, given the parameters $\eta$, we can compute the propensity score for any neighborhood configuration $\mathbf{a}_{\cN_k}$ via
\begin{align}\label{eqn:simu-est1}
\begin{aligned}
&\mP(\bA_{\cN_k}=\mathbf{a}_{\cN_k}|\mathbf{L},\{A_j,j\notin \cN_k\}) = \frac{\mP(\bA_{\cN_k}=\mathbf{a}_{\cN_k},\{A_j,j\notin \cN_k\}|\mathbf{L})}{\sum_{\mathbf{a}^*}P(\bA_{\cN_k}=\mathbf{a}^*,\{A_j,j\notin \cN_k\}|\mathbf{L})
}\\
&=\frac{\exp\{\sum_{i\in \cN_k} a_iG_i(\mathbf{L})+\eta_7\sum_{i\in \cN_k}  \sum_{j\in \cN(i;1)\cap \cN_k}a_ia_j+\eta_7\sum_{i\in \cN_k}  \sum_{j\in \cN(i;1)\setminus \cN_k}a_iA_j\}}{\sum_{\mathbf{a}^*}\exp\{\sum_{i\in \cN_k} a_i^*G_i(\mathbf{L})+\eta_7\sum_{j\in \cN(i;1)\cap \cN_k}a_i^*a_j^*+\eta_7\sum_{i\in \cN_k} \sum_{j\in \cN(i;1)\setminus \cN_k}a^*_iA_j\}}~.    
\end{aligned}
\end{align}
To construct $\widehat W_i$, we only need to evaluate the propensity score for observed neighborhood treatment configurations,
\begin{align*}
&\mP(\bA_{\cN_k}=\mathbf{A}_{\cN_k}^{obs}|\mathbf{L},\{A_j^{obs},j\notin \cN_k\})~.
\end{align*}
Substituting the observed values $\mathbf{A}^{\mathrm{obs}}$ for $A_j$ and $a_i$ in \eqref{eqn:simu-est1}, together with $\mathbf{L}^{\mathrm{obs}}$ and the estimated parameter $\widehat\eta$ from Logistic regressions, yields the numerator. For the denominator, we substitute $A_j^{\mathrm{obs}}$ for $A_j$ in \eqref{eqn:simu-est1} and sum over all possible configurations of $a_i^*$. In practice, directly evaluating these exponentials in \eqref{eqn:simu-est1} might be computationally infeasible. In this case, the propensity score can be approximated using Gibbs sampling rather than the explicit formula in \eqref{eqn:simu-est1}. When adopting the exposure mapping $T_i = T(\mathbf{a}_{\mathcal{N}_i})$ introduced in Section~\ref{subsec:exposure}, we consider another exposure-propensity score, $\mP(T_i = t \mid \mathbf{L}, \mathbf{A}_{-\mathcal{N}_i})$, which can be computed as
\begin{equation*}
    \mP(T_i = t \mid \mathbf{L}, \mathbf{A}_{-\cN_i}) = \sum_{\mathbf{a}_{\cN_i}} \mP(\mathbf{A}_{\cN_i} = \mathbf{a}_{\cN_i} \mid \mathbf{L}, \mathbf{A}_{-\cN_i}) 1\{T(\mathbf{a}_{\cN_i}) = t\} ~.
\end{equation*}
With $\beta_{\mathbf{a}_{\mathcal{N}_i}}(X_i^y; \widehat\theta)$ and $\pi_{\mathbf{a}_{\mathcal{N}_i}}(X_i^a; \widehat\eta)$ in place, we can compute $\widehat W_i$ and subsequently construct $\widehat \mu_{\mathrm{AIPW}}$. Recall that $\widehat \mu_{\mathrm{AIPW}}$ provides an estimate of $\mE[Y_i(\mathbf{a})]$; using this quantity, the estimands listed in Table~\ref{table:estimands} can then be obtained.

\subsection{Extra Simulation Results}\label{app:extra_sim}
In this section, we present additional simulation results. We evaluate the performance of our proposed AIPW estimators under three scenarios: (i) both the propensity and outcome models are correctly specified; (ii) only the propensity model is correctly specified; and (iii) only the outcome model is correctly specified. Results for the Auto-G-Computation estimators are omitted, as the method yields identical results under Scenarios (i) and (iii), and Scenario (iii) has already been reported in Table~\ref{table:aipw_autog_all}.

\begin{table}[htbp]
\small
\renewcommand{\arraystretch}{0.75}
\centering
\begin{tabular}{
  lccccccccc
}
\toprule
 &
\multicolumn{3}{c}{\textbf{Scenario (i)}} &
\multicolumn{3}{c}{\textbf{Scenario (ii)}} &
\multicolumn{3}{c}{\textbf{Scenario (iii)}} \\
\cmidrule(lr){2-4} \cmidrule(lr){5-7} \cmidrule(lr){8-10}
($m$,\texttt{max\_degree}) & {(1,2)} & {(2,5)} & {(3,10)} & {(1,2)} & {(2,5)} & {(3,10)} & {(1,2)} & {(2,5)} & {(3,10)} \\
\midrule
\multicolumn{10}{l}{$\gamma(\alpha)$} \\
AAIPW Bias    & 0.002 & 0.001 & 0.000 & 0.011 & -0.001 & -0.005 & 0.002 & 0.001 & -0.000 \\
AAIPW RMSE    & 0.046 & 0.032 & 0.022 & 0.052 & 0.038 & 0.027 & 0.024 & 0.017 & 0.012 \\
\midrule
\multicolumn{10}{l}{ $\DE(\alpha)$ } \\
AAIPW Bias    & 0.002 & 0.002 & 0.001 & 0.016 & -0.005 & -0.019 & 0.002 & 0.002 & 0.001 \\
AAIPW RMSE    & 0.069 & 0.048 & 0.034 & 0.078 & 0.058 & 0.048 & 0.036 & 0.027 & 0.020 \\
\midrule
\multicolumn{10}{l}{$\IE(\alpha, \alpha')$} \\
AAIPW Bias    & 0.004 & 0.003 & 0.004 & 0.009 & 0.020 & 0.035 & 0.004 & 0.004 & 0.004 \\
AAIPW RMSE    & 0.031 & 0.045 & 0.079 & 0.038 & 0.059 & 0.099 & 0.018 & 0.027 & 0.053 \\
\midrule
\multicolumn{10}{l}{$\IE(\alpha)$ } \\
AAIPW Bias    & 0.001 & 0.006 & 0.006 & 0.013 & 0.070 & 0.149 & 0.001 & 0.006 & 0.003 \\
AAIPW RMSE    & 0.064 & 0.159 & 0.211 & 0.081 & 0.210 & 0.274 & 0.037 & 0.096 & 0.159 \\
\bottomrule
\end{tabular}
\caption{Estimation accuracy of AAIPW across different causal parameters, network densities, and scenarios. Results are averaged over 8,000 Monte Carlo samples generated via Gibbs sampling.}
\label{table:aipw_all}
\end{table}
An interesting observation is that our AAIPW estimator performs better under Scenario (iii) than under Scenario (i). In Scenario (iii), the inputs to the propensity score model consist largely of random noise, resulting in estimated coefficients from logistic regression that are close to zero. This, in turn, leads to more stable inverse propensity score weights with smaller variance. Since the AAIPW estimator remains consistent under both scenarios, a reduction in variance translates directly into a lower RMSE. However, it is important to note that using an uninformative propensity model can lead to poor performance if the outcome model is also misspecified. In this sense, the double robustness property involves a trade-off, as the lower variance under Scenario (iii) may come at the cost of increased sensitivity when both models are imperfect.

\section{Additional Theoretical Result}\label{app:extra-identification}
\subsection{Regularization Assumptions}
In this section, we present the regularity assumptions required to establish our asymptotic results.
\begin{assumption}\label{ass:regularization}
Following regularization assumptions hold:
\begin{enumerate}[label=(\roman*)]
    \item $\beta_{\mathbf{a}_{\cN_i}}(\cdot;\theta)$ and $\pi_{\mathbf{a}_{\cN_i}}(\cdot;\eta)$ are differentiable with respect to $\theta$ and $\eta$, respectively, with probability 1. Moreover, $\beta_{\mathbf{a}_{\cN_i}}(\cdot;\theta)$, $\pi_{\mathbf{a}_{\cN_i}}(\cdot;\eta)$, and their derivatives $\partial \beta_{\mathbf{a}_{\cN_i}}(\cdot;\theta)/\partial \theta$ and $\partial \pi_{\mathbf{a}_{\cN_i}}(\cdot;\eta)/\partial \eta$ are Lipschitz continuous.
    \item There exists some constant $c_1$ and $c_2$, such that for any unit $i$, $c_1\le\pi_{\mathbf{a}_{\cN_i}}(X_i^a;\eta)\le c_2$ with probability 1.

    \item The following limitations exist,
\begin{align*}
   &M_1: =\lim_{n\to\infty}\frac{1}{n}\sum_{i=1}^n\mathrm{\mE}\left[ H_{1,i}\right]~,\quad M_2: =\lim_{n\to\infty}\frac{1}{n}\sum_{i=1}^n\mathrm{\mE}\left[H_{2,i}\right]~,
\end{align*}
where $H_{1,i}$ and $H_{2,i}$ is given by
\begin{align}\label{eqn:H12-define}
\begin{aligned}
    &H_{1,i}: =   \left(1 - \frac{1\{ \mathbf{A}_{\cN_i} = \mathbf{a}_{\cN_i}\} }{\pi_{\mathbf{a}_{\cN_i}}(X_i^a;\eta)} \right) \frac{\partial \beta_{\mathbf{a}_{\cN_i}}(X_i^y;\theta)}{\partial \theta}\\
    &H_{2,i}:= \frac{1\{ \mathbf{A}_{\cN_i} = \mathbf{a}_{\cN_i}\} }{\left(\pi_{\mathbf{a}_{\cN_i}}(X_i^a;\eta)\right)^2} \frac{\partial \pi_{\mathbf{a}_{\cN_i}}(X_i^a;\eta)}{\partial \eta} \left(Y_i-\beta_{\mathbf{a}_{\cN_i}}(X_i^y;\theta)\right)
\end{aligned}
\end{align}

\item In the case that $\widehat\theta$ and $\widehat\eta$ are obtained by solving sample moment equations, $\frac{1}{n}\sum_{i=1}^n g^y(Y_i,X_i^y;\widehat\theta) = 0$ and $\frac{1}{n}\sum_{i=1}^n g^a(Y_i,X_i^a;\widehat\eta) = 0$, the following limitations exist
\begin{align*}
    &M_3:=\lim_{n\to\infty}\frac{1}{n}\sum_{i=1}^n\mE\left[\frac{\partial g^y(Y_i,X_i^y;\theta)}{\partial\theta}\right]\\
     &M_4: = \lim_{n\to\infty}\frac{1}{n}\sum_{i=1}^n\mE\left[\frac{\partial g^a(Y_i,X_i^a;\eta)}{\partial\eta}\right]~.
\end{align*}

 \item For any unit $i$, the fourth moments of $H_{1,i}$, $H_{2,i}$, $\partial g^y(Y_i,X_i^y;\theta)/\partial\theta$ and $\partial g^a(Y_i,X_i^a;\eta)/\partial\eta$ exist. 
 
\item Let 
\[
Z_i := (Y_j, A_j, L_j : j \in \mathcal{N}_n(i, K+1)).
\]  
Consider $H_{1,i}$, $H_{2,i}$, $\partial g^y(Y_i,X_i^y;\theta)/\partial\theta$, and $\partial g^a(Y_i,X_i^a;\eta)/\partial\eta$ as measurable functions of $Z_i$. We assume that each of these functions is \emph{locally Lipschitz} in $Z_i$, in the sense that there exists a constant $C_n = n^q$ for some $q>0$ and a function $D$ such that, for all $Z', Z''$ in $[-n, n]$,
\[
|f(Z') - f(Z'')| \le C_n \, D(Z' - Z''), 
\quad f \in \{ H_{1,i}, H_{2,i}, \partial g^y/\partial\theta, \partial g^a/\partial\eta \}.
\]

\end{enumerate}
\end{assumption}

\subsection{Additional Identification Result}
In this section, we introduce the following weaker form of network conditional ignorability and derive a new identification result based on it, as discussed in footnote \ref{fn:weak-identification}.
\begin{assumption}[Network Conditional Ignorability II]\label{ass:ignorability2}
\begin{equation*}
    \mathbf{A} \independent Y_i(\mathbf{a}) \mid \mathbf{L} \quad \text{for all } \mathbf{a} \in \mathscr{A}(n)~.
\end{equation*}
\end{assumption}
Compared to Assumption~\ref{ass:ignorability}, for a given unit $i$, Assumption~\ref{ass:ignorability2} requires conditional independence only between $\mathbf{A}$ and $Y_i(\mathbf{a})$, without imposing restrictions on the outcomes of other units. Despite this weaker condition, it still suffices to identify $\mE\left[ Y_i(\mathbf{a}_{\cN_i}, \mathbf{A}_{-\cN_i}) \right]$ for unit $i$. The following theorem establishes this identification result, relying on an alternative NAIPW formulation with an outcome regression that excludes $\mathbf{Y}_{-i}$. The proof is provided in Section~\ref{app:proofs}.
\begin{theorem}\label{thm:identification_extra}
    Under Assumption \ref{ass:consistency}, \ref{ass:positivity} and \ref{ass:ignorability2}, for any fixed distance $K\ge 1$ we have:
    \begin{equation}\label{eqn:identification1-2}
    \begin{aligned}
        &\mE\left[  \frac{1\left\{ \mathbf{A}_{\cN_i} = \mathbf{a}_{\cN_i} \right\} \left(Y_i - \mE[Y_i \mid \mathbf{A}_{\cN_i}=\mathbf{a}_{\cN_i}, \mathbf{A}_{-\cN_i}, \mathbf{L}]\right)}{\mP\left( \mathbf{A}_{\cN_i} = \mathbf{a}_{\cN_i} \mid \mathbf{L}, \mathbf{A}_{-\cN_i} \right)} + \mE[Y_i \mid \mathbf{A}_{\cN_i}=\mathbf{a}_{\cN_i}, \mathbf{A}_{-\cN_i}, \mathbf{L}] \right] \\
        &= \mE\left[ Y_i(\mathbf{a}_{\cN_i}, \mathbf{A}_{-\cN_i}) \right]~,
    \end{aligned}
    \end{equation}
    if either $\mP\left( \mathbf{A}_{\cN_i} = \mathbf{a}_{\cN_i} \mid \mathbf{L}, \mathbf{A}_{-\cN_i} \right)$ or $\mE[Y_i \mid \mathbf{A}_{\cN_i}=\mathbf{a}_{\cN_i}, \mathbf{A}_{-\cN_i}, \mathbf{L}]$ is correctly specified.
\end{theorem}
\subsection{Convergence of Parameter Estimators}
As discussed in Section \ref{subsec:main2}, in this section, we present a proposition demonstrating that the GMM-based estimator satisfies Assumption \ref{ass:estimator_consistent} under the weak dependence condition and the regularity assumptions. Formally, we have
\begin{proposition}\label{prop:gmm-consistency}
Under Assumption \ref{ass:network_weak_dependence} and \ref{ass:regularization}, $\widehat\theta \overset{a.s.}{\longrightarrow} \theta$ and $\widehat\eta \overset{a.s.}{\longrightarrow} \eta$.
\end{proposition}

\section{Mathematical Proofs}\label{app:proofs}
\subsection{Proof of Theorem \ref{thm:identification1}}
\begin{proof}
Throughout this proof, we use $\widetilde \mE[Y_i \mid \mathbf{Y}_{-i}, \mathbf{A}_{\cN_i}=\mathbf{a}_{\cN_i}, \mathbf{A}_{-\cN_i}, \mathbf{L}]$ and $\widetilde\mP(\mathbf{A}_{\cN_i} = \mathbf{a}_{\cN_i} \mid \mathbf{L}, \mathbf{A}_{-\cN_i})$ to denote the (possibly misspecified) conditional outcome and the propensity score. With this notation, the left-hand side of equation~\eqref{eqn:identification1} can be decomposed into three components:
\begin{align}
\label{eqn:thm1_0}
\begin{aligned}
\text{LHS of \eqref{eqn:identification1}} = &\mE[\widetilde\mE[Y_i \mid \mathbf{Y}_{-i}, \mathbf{A}_{\cN_i}=\mathbf{a}_{\cN_i}, \mathbf{A}_{-\cN_i}, \mathbf{L}]]  + \mE\left[ \frac{1\{\mathbf{A}_{\cN_i} = \mathbf{a}_{\cN_i}\} Y_i}{\widetilde\mP(\mathbf{A}_{\cN_i} = \mathbf{a}_{\cN_i} \mid \mathbf{L}, \mathbf{A}_{-\cN_i})} \right] \\
&-\mE\left[ \frac{1\{\mathbf{A}_{\cN_i} = \mathbf{a}_{\cN_i}\} \widetilde\mE[Y_i \mid \mathbf{Y}_{-i}, \mathbf{A}_{\cN_i}=\mathbf{a}_{\cN_i}, \mathbf{A}_{-\cN_i}, \mathbf{L}]}{\widetilde\mP(\mathbf{A}_{\cN_i} = \mathbf{a}_{\cN_i} \mid \mathbf{L}, \mathbf{A}_{-\cN_i})} \right] \\
= &: \cD_1 +  \cD_2 -  \cD_3.
\end{aligned}
\end{align}
We next consider two cases separately: either the conditional outcome model or the propensity score is correctly specified. In each case, we show that \eqref{eqn:thm1_0} equals to $\mE[Y_i(\mathbf{a}_{\cN_i}, \mathbf{A}_{-\cN_i})]$, demonstrating that the identification result holds if at least one of the two models is correctly specified.
\subsubsection*{Case 1: $\widetilde\mE[Y_i \mid \mathbf{Y}_{-i}, \mathbf{A}_{\cN_i}=\mathbf{a}_{\cN_i}, \mathbf{A}_{-\cN_i}, \mathbf{L}] = \mE[Y_i \mid \mathbf{Y}_{-i}, \mathbf{A}_{\cN_i}=\mathbf{a}_{\cN_i}, \mathbf{A}_{-\cN_i}, \mathbf{L}]~.$}

Given that the conditional expectation of $Y_i$ is correctly specified in this case, we directly have $\cD_1 = \mE\big[\mE[Y_i(\mathbf{a}_{\cN_i}, \mathbf{A}_{-\cN_i}) \mid \mathbf{Y}_{-i}, \mathbf{A}_{\cN_i}=\mathbf{a}_{\cN_i}, \mathbf{A}_{-\cN_i}, \mathbf{L}]\big]$. In addition, Lemma \ref{lemma:independent}(ii) implies that the conditional distributions $Y_i(\mathbf{a}_{\cN_i}, \mathbf{A}_{-\cN_i}) \mid \mathbf{Y}_{-i}, \mathbf{A}_{\cN_i}=\mathbf{a}_{\cN_i}, \mathbf{A}_{-\cN_i}, \mathbf{L}$ and $Y_i(\mathbf{a}_{\cN_i}, \mathbf{A}_{-\cN_i}) \mid \mathbf{Y}_{-i}(\mathbf{a}_{\cN_i}, \mathbf{A}_{-\cN_i}),\mathbf{A}_{-\cN_i}, \mathbf{L}$ coincide. Consequently, $\cD_1$ can be written as
\begin{align}\label{eqn:thm1_1}
\begin{aligned}
\cD_1 &= \mE\big[\mE[Y_i(\mathbf{a}_{\cN_i}, \mathbf{A}_{-\cN_i}) \mid \mathbf{Y}_{-i}(\mathbf{a}_{\cN_i}, \mathbf{A}_{-\cN_i}), \mathbf{A}_{-\cN_i}, \mathbf{L}]\big] = \mE[Y_i(\mathbf{a}_{\cN_i}, \mathbf{A}_{-\cN_i})].
\end{aligned}
\end{align}
We then consider $\cD_2$ and $\cD_3$. It turns out that in Case 1, we have
\begin{align*}
     \cD_2  &= \mE\left[\mE\left[\frac{1\{\mathbf{A}_{\cN_i} = \mathbf{a}_{\cN_i} \}Y_i(\mathbf{a}_{\cN_i}, \mathbf{A}_{-\cN_i})}{\widetilde\mP(\mathbf{A}_{\cN_i} = \mathbf{a}_{\cN_i} \mid \mathbf{L}, \mathbf{A}_{-\cN_i} ) } \mid \mathbf{Y}_{-i}, \mathbf{A}_{\cN_i}=\mathbf{a}_{\cN_i}, \mathbf{A}_{-\cN_i}, \mathbf{L}\right]\right]\\
    &= \mE\left[\frac{ 1\{\mathbf{A}_{\cN_i} = \mathbf{a}_{\cN_i} \} \mE[Y_i(\mathbf{a}_{\cN_i}, \mathbf{A}_{-\cN_i}) \mid \mathbf{Y}_{-i}, \mathbf{A}_{\cN_i}=\mathbf{a}_{\cN_i}, \mathbf{A}_{-\cN_i}, \mathbf{L}]}{\widetilde\mP(\mathbf{A}_{\cN_i} = \mathbf{a}_{\cN_i} \mid \mathbf{L}, \mathbf{A}_{-\cN_i} ) }  \right] \hspace{6mm}  \\
    &=\mE\left[ \frac{1\{\mathbf{A}_{\cN_i} = \mathbf{a}_{\cN_i}\} \mE[Y_i \mid \mathbf{Y}_{-i}, \mathbf{A}_{\cN_i}=\mathbf{a}_{\cN_i}, \mathbf{A}_{-\cN_i}, \mathbf{L}]}{\widetilde\mP(\mathbf{A}_{\cN_i} = \mathbf{a}_{\cN_i} \mid \mathbf{L}, \mathbf{A}_{-\cN_i})} \right] = \cD_3~.
\end{align*}
Therefore, $\cD_2 = \cD_3$ in Case 1. Combing this equation with \eqref{eqn:thm1_1}, we conclude that \eqref{eqn:thm1_0} equals $\mE[Y_i(\mathbf{a}_{\cN_i}, \mathbf{A}_{-\cN_i})]$ in Case 1.

\subsubsection*{Case 2: $\widetilde\mP(\mathbf{A}_{\cN_i} = \mathbf{a}_{\cN_i} \mid \mathbf{L}, \mathbf{A}_{-\cN_i}) = \mP(\mathbf{A}_{\cN_i} = \mathbf{a}_{\cN_i} \mid \mathbf{L}, \mathbf{A}_{-\cN_i})~.$}
With Assumption \ref{ass:consistency}$, \cD_2$ can be written as
    \begin{align}\label{eqn:thm1_4}
    \begin{aligned}
        \cD_2 =& \mE\left[\mE\left[  \frac{1\{\mathbf{A}_{\cN_i} = \mathbf{a}_{\cN_i} \} Y_i}{\mP(\mathbf{A}_{\cN_i} = \mathbf{a}_{\cN_i} \mid \mathbf{L}, \mathbf{A}_{-\cN_i} )}   \mid \mathbf{L},\mathbf{A}_{-\cN_i} \right] \right]\\
        =&\mE\left[  \frac{ \mE\left[1\{\mathbf{A}_{\cN_i} = \mathbf{a}_{\cN_i} \} Y_i \mid \mathbf{L},\mathbf{A}_{-\cN_i}\right]}{\mP(\mathbf{A}_{\cN_i} = \mathbf{a}_{\cN_i} \mid \mathbf{L}, \mathbf{A}_{-\cN_i} )}  \right]\\
        =& \mE\left[  \frac{ \mE\left[1\{\mathbf{A}_{\cN_i} = \mathbf{a}_{\cN_i} \} Y_i(\mathbf{a}_{\cN_i}, \mathbf{A}_{-\cN_i}) \mid \mathbf{L},\mathbf{A}_{-\cN_i}\right]}{\mP(\mathbf{A}_{\cN_i} = \mathbf{a}_{\cN_i} \mid \mathbf{L}, \mathbf{A}_{-\cN_i} )}  \right]~.
    \end{aligned}
    \end{align}
    Using Lemma \ref{lemma:independent}(i), we have $\mathbf{Y}(\mathbf{a}_{\cN_i}, \mathbf{a}_{-\cN_i}) \independent \mathbf{A}_{\cN_i} \mid \mathbf{A}_{-\cN_i}, \mathbf{L}$. Consequently, 
    \begin{align} \label{eqn:thm1_6}
        \begin{aligned}
            \cD_2 =\mE\left[  \frac{ \mE\left[1\{\mathbf{A}_{\cN_i} = \mathbf{a}_{\cN_i} \} Y_i(\mathbf{a}_{\cN_i}, \mathbf{A}_{-\cN_i}) \mid \mathbf{L},\mathbf{A}_{-\cN_i}\right]}{\mP(\mathbf{A}_{\cN_i} = \mathbf{a}_{\cN_i} \mid \mathbf{L}, \mathbf{A}_{-\cN_i} )}  \right] &= \mE\left[\mE\left[ Y_i(\mathbf{a}_{\cN_i}, \mathbf{A}_{-\cN_i}) \mid \mathbf{L},\mathbf{A}_{-\cN_i}\right] \right] \\
            &= \mE\left[ Y_i(\mathbf{a}_{\cN_i}, \mathbf{A}_{-\cN_i})\right] ~.
        \end{aligned}
    \end{align}
For $\cD_3$, Lemma \ref{lemma:independent}(i) implies that $\widetilde\mE[Y_i \mid \mathbf{Y}_{-i}(\mathbf{a}_{\cN_i}, \mathbf{A}_{-\cN_i}),\mathbf{A}_{\cN_i} = \mathbf{a}_{\cN_i}, \mathbf{A}_{-\cN_i},\mathbf{L}]$ and $\mathbf{A}_{\cN_i}$ are independent conditioning on $\mathbf{L}$ and $\mathbf{A}_{-\cN_i}$. Therefore,
    \begin{align}
    \label{eqn:thm1_8}
    \begin{aligned}
        \cD_3=& \mE\left[\mE\left[  \frac{1\{\mathbf{A}_{\cN_i} = \mathbf{a}_{\cN_i} \} \widetilde\mE[Y_i \mid \mathbf{Y}_{-i}(\mathbf{a}_{\cN_i}, \mathbf{A}_{-\cN_i}), \mathbf{A}_{\cN_i} = \mathbf{a}_{\cN_i}, \mathbf{A}_{-\cN_i},\mathbf{L}]}{\mP(\mathbf{A}_{\cN_i} = \mathbf{a}_{\cN_i} \mid \mathbf{L}, \mathbf{A}_{-\cN_i} )}  \mid \mathbf{L}, \mathbf{A}_{-\cN_i} \right]\right]\\
        =&\mE\left[\frac{\mE[\widetilde\mE[Y_i \mid \mathbf{Y}_{-i}(\mathbf{a}_{\cN_i}, \mathbf{A}_{-\cN_i}), \mathbf{A}_{\cN_i} = \mathbf{a}_{\cN_i}, \mathbf{A}_{-\cN_i},\mathbf{L}]\mid\mathbf{A}_{-\cN_i},\mathbf{L}]}{\mP(\mathbf{A}_{\cN_i} = \mathbf{a}_{\cN_i} \mid \mathbf{L}, \mathbf{A}_{-\cN_i} ) }\mE\left[  1\{\mathbf{A}_{\cN_i} = \mathbf{a}_{\cN_i} \}  \mid \mathbf{L}, \mathbf{A}_{-\cN_i} \right]\right]\\
     =&\mE\left[\mE[\widetilde\mE[Y_i \mid \mathbf{Y}_{-i}(\mathbf{a}_{\cN_i}, \mathbf{A}_{-\cN_i}), \mathbf{A}_{\cN_i} = \mathbf{a}_{\cN_i}, \mathbf{A}_{-\cN_i},\mathbf{L}]\mid\mathbf{A}_{-\cN_i},\mathbf{L}]\right] \\
     =& \mE[\widetilde\mE[Y_i \mid \mathbf{Y}_{-i}, \mathbf{A}_{\cN_i}=\mathbf{a}_{\cN_i}, \mathbf{A}_{-\cN_i}, \mathbf{L}]] = \cD_1~.
    \end{aligned}
    \end{align}
Combing \eqref{eqn:thm1_6} and \eqref{eqn:thm1_8}, we conclude that \eqref{eqn:thm1_0} equals $\mE[Y_i(\mathbf{a}_{\cN_i}, \mathbf{A}_{-\cN_i})]$ in Case 2.
\end{proof}
\subsection{Proof of Proposition \ref{prop:models} and  Proposition \ref{thm:identification1-markov}}
\begin{proof}
    This proposition could be proved by Markov properties in Assumption \ref{ass:markov}. Specifically, Assumption \ref{ass:markov}(i) gives that for all $i$ and  $k \notin \cN_n(i;1)$, we have
        \begin{align*}
            Y_i \independent \{Y_k, A_k, L_k\} \mid \left(A_i, L_i, \{Y_j, A_j, L_j : j \in \cN_n^\partial(i;1)\}\right)~.
        \end{align*}
    Therefore, we immediately have
    \begin{equation*}
        \mE[Y_i \mid \mathbf{Y}_{-i}, \mathbf{A} = \mathbf{a}, \mathbf{L}] = \mE[Y_i \mid \mathbf{A}_{\cN_i} = \mathbf{a}_{\cN_i}, L_i, \{Y_j, L_j : j \in \cN^\partial_n(i;1)\}]~.
    \end{equation*}

Similarly,  Assumption \ref{ass:markov}(ii) gives 
\begin{align*}
    \mathbf{A}_{\cN_i} &\independent \{A_k, L_k\} \mid \{L_j : j \in \cN_n(i, K+1)\}, \{A_j : j \in \cN_n^\partial(i; K+1)\}  \text{ for all } i \text{ and } k \notin \cN_n(i;K+1)~,
\end{align*}
which immediately leads to
 \begin{equation*}
        \mP\left(\mathbf{A}_{\cN_i} = \mathbf{a}_{\cN_i} \mid \mathbf{L}, \mathbf{A}_{-\cN_i}\right) = \mP(\mathbf{A}_{\cN_i} = \mathbf{a}_{\cN_i} \mid \{L_j : j \in \cN_n(i, K+1)\}, \{A_j : j \in \cN_n^\partial(i; K+1)\})~.
    \end{equation*}
This completes the proof of Proposition \ref{prop:models}. Together with Theorem \ref{thm:identification1}, we also have Proposition \ref{thm:identification1-markov} holds.

\end{proof}

\subsection{Proof of Theorem \ref{thm:identification2} and Proposition \ref{thm:identification2-markov}}
\begin{proof}
In view of Theorem~\ref{thm:identification1}, the proof of Theorem~\ref{thm:identification2} reduces to verifying that
\begin{align}\label{eqn:thm2_1}
    \mE\left[ Y_i(\mathbf{a}_{\cN_i}, \mathbf{A}_{-\cN_i})   \right] =\mE[Y_i(\mathbf{a})]~.
\end{align} 
We next consider two cases separetely: either the the neighborhood interference condition stated in Assumption~\ref{ass:interference} holds or the Markov property in Assumption~\ref{ass:markov}(i) holds.
\subsubsection*{Case 1: Assumption~\ref{ass:interference} holds.}
Under the neighborhood interference condition stated in Assumption~\ref{ass:interference}, $Y_i(\mathbf{a}_{\cN_i}, \mathbf{a}^*_{-\cN_i}) = Y_i(\mathbf{a})$ for any $\mathbf{a}^*_{-\cN_i}$. Consequently, \eqref{eqn:thm2_1} follows directly, as
\begin{align*}
    \mE\left[ Y_i(\mathbf{a}_{\cN_i}, \mathbf{A}_{-\cN_i})   \right] &= \sum_{\mathbf{a}^*_{-\cN_i}}\mE\left[Y_i(\mathbf{a}_{\cN_i},\mathbf{a}^*_{-\cN_i})1 \{\mathbf{A}_{-\cN_i}=\mathbf{a}^*_{-\cN_i}\}\right]\\&= \sum_{\mathbf{a}^*_{-\cN_i}}\mE\left[Y_i(\mathbf{a})1 \{\mathbf{A}_{-\cN_i}=\mathbf{a}^*_{-\cN_i}\}\right]=\mE[Y_i(\mathbf{a})]~,
\end{align*}
which verifies \eqref{eqn:thm2_1}.
\subsubsection*{Case 2: Assumption~\ref{ass:markov}(i) holds.}
From the proof of Theorem~\ref{thm:identification1}, specifically equation \eqref{eqn:thm1_1}, we have
\begin{align*}
    \mE\left[ Y_i(\mathbf{a}_{\cN_i}, \mathbf{A}_{-\cN_i})   \right] &= \mE\left[\mE[Y_i \mid \mathbf{Y}_{-i}, \mathbf{A}_{\cN_i}=\mathbf{a}_{\cN_i}, \mathbf{A}_{-\cN_i}, \mathbf{L}]  \right]~. 
\end{align*}
Assumption~\ref{ass:markov}(i) implies that the conditioning set can be reduced, yielding
\begin{align*}
    \mE\left[\mE[Y_i \mid \mathbf{Y}_{-i}, \mathbf{A}_{\cN_i}=\mathbf{a}_{\cN_i}, \mathbf{A}_{-\cN_i}, \mathbf{L}]  \right] =\mE\left[    \mE[Y_i \mid \mathbf{Y}_{-i}, \mathbf{A}=\mathbf{a}, \mathbf{L}] \right]~.  
\end{align*}
By Lemma~\ref{lemma:independent}(iii), the conditional distributions $Y_i \mid \mathbf{Y}_{-i}, \mathbf{A}=\mathbf{a}, \mathbf{L}$ coincides with that of $Y_i(\mathbf{a}) \mid \mathbf{Y}_{-i}(\mathbf{a}), \mathbf{L}$. Consequently,
\begin{equation} \label{eq:app-proof-thm4-1}
    \begin{aligned}
        \mE\left[ Y_i(\mathbf{a}_{\cN_i}, \mathbf{A}_{-\cN_i})   \right] 
        &= \mE\left[    \mE[Y_i(\mathbf{a}) \mid \mathbf{Y}_{-i}(\mathbf{a}), \mathbf{L}] \right]  
        = \mE[Y_i(\mathbf{a})]~.
    \end{aligned}
\end{equation}
This establishes \eqref{eqn:thm2_1} and completes the proof of Theorem~\ref{thm:identification2}. Combining this result with Proposition~\ref{thm:identification1-markov} yields Proposition~\ref{thm:identification2-markov}.
\end{proof}
\subsection{Proof of Proposition \ref{prop:gmm-consistency}}\label{app:proof-gmm-consistency}
\begin{proof}
      We consider the setting where the outcome and propensity score model estimators are obtained via GMM. The  moment conditions used are $\frac{1}{n}\sum_{i=1}^n g^y(Y_i, X_i^y;\widehat\theta)=0$ and $\frac{1}{n}\sum_{i=1}^n g^a(Y_i, X_i^a;\widehat\eta)=0$. In this case, the estimator $\widehat\theta$ satisfies
    \begin{align*}
        \sqrt{n}(\widehat{\theta}-\theta)=-\left[\frac{1}{n}\sum_{i=1}^n\frac{\partial g^y(Y_i,X_i^y;\theta)}{\partial\theta}\right]^{-1}\frac{1}{\sqrt{n}}\sum_{i=1}^ng^y(Y_i,X_i^y;\theta)+o_p(1).
    \end{align*}
    Lemma \ref{lemma:M-consistency} shows that, with Assumption \ref{ass:network_weak_dependence} and \ref{ass:regularization}, we have 
    \begin{align}\label{eqn:m3}
        \frac{1}{n}\sum_{i=1}^n\frac{\partial g^y(Y_i,X_i^y;\theta)}{\partial\theta} \overset{a.s.}{\longrightarrow} \lim_{n\to\infty}\frac{1}{n}\sum_{i=1}^n\mE\left[\frac{\partial g^y(Y_i,X_i^y;\theta)}{\partial\theta}\right]: = M_3~.
    \end{align}
    Similar equation could be derived for $\widehat\eta$:
    \begin{align}\label{eqn:m4}
        \frac{1}{n}\sum_{i=1}^n\frac{\partial g^a(Y_i,X_i^a;\eta)}{\partial\eta} \overset{a.s.}{\longrightarrow} \lim_{n\to\infty}\frac{1}{n}\sum_{i=1}^n\mE\left[\frac{\partial g^a(Y_i,X_i^a;\eta)}{\partial\eta}\right]: = M_4~.
    \end{align}
Consequently, we have
    \begin{align}
        \sqrt{n}(\widehat \theta - \theta) = - M_3^{-1}  \frac{1}{\sqrt{n}}\sum_{i=1}^n g^y(Y_i, X_i^y;\theta) + o_p(1) ~,
    \end{align}
  and
    \begin{align}
        \sqrt{n}(\widehat \eta - \eta) = -M_4^{-1} \frac{1}{\sqrt{n}}\sum_{i=1}^n g^a(Y_i, X_i^a;\theta) + o_p(1) ~,
    \end{align}
     Assumptions~\ref{ass:network_weak_dependence} and~\ref{ass:regularization} guarantee the network weak dependence of $g^y(Y_i, X_i^y;\theta)$ and $g^a(Y_i, X_i^a;\theta)$. Following \citet{kojevnikov2021limit}, we have the consistency result directly.
\end{proof}

\subsection{Proof of Theorem \ref{thm:clt}}\label{app:proof-clt}
\begin{proof}
The AAIPW estimator can be written as 
\begin{align}
    \begin{aligned}\label{clt:eqn1}
        &\sqrt{n}\left(\widehat \mu_{\mathrm{AIPW}}-\mu_n \right) \\=& \frac{1}{\sqrt{n}} \sum_{i=1}^n \frac{1\{ \mathbf{A}_{\cN_i}
        = \mathbf{a}_{\cN_i}\} }{\pi_{\mathbf{a}_{\cN_i}}(X_i^a;\widehat\eta)} \left(Y_i-\beta_{\mathbf{a}_{\cN_i}}(X_i^y;\widehat\theta)\right) +\beta_{\mathbf{a}_{\cN_i}}(X_i^y;\widehat\theta) -  \mu_i ~,
    \end{aligned}
\end{align}  
using a first-order Taylor approximation around $\theta$ and $\eta$ of the propensity and outcome models, we have
    \begin{align}
    \begin{aligned}\label{clt:eqn2}
         &\sqrt{n}\left(\widehat \mu_{\mathrm{AIPW}}-\mu_n \right) \\=& \frac{1}{\sqrt{n}}\sum_{i=1}^n \widetilde W_{i} +  \left(\frac{1}{n}\sum_{i=1}^n H_{1,i}\right)\cdot \sqrt{n}(\widehat\theta - \theta)  - \left(\frac{1}{n}\sum_{i=1}^n H_{2,i}\right)\cdot \sqrt{n}(\widehat\eta - \eta) + o_p(1)~, 
    \end{aligned}
    \end{align}
    where $\widetilde W_{i}$ is given by
\begin{align*}
     \widetilde W_{i} &=  \frac{1\{ \mathbf{A}_{\cN_i} = \mathbf{a}_{\cN_i}\} }{\pi_{\mathbf{a}_{\cN_i}}(X_i^a;\eta)} \left(Y_i-\beta_{\mathbf{a}_{\cN_i}}(X_i^y;\theta)\right) +\beta_{\mathbf{a}_{\cN_i}}(X_i^y;\theta) -  \mu_i~.
\end{align*}
$H_{1,i}$ and $H_{2,i}$ is given by \eqref{eqn:H12-define}. Using Lemma \ref{lemma:M-consistency}, \eqref{clt:eqn2} can be further written as
\begin{align}\label{eqn:clt_eqn3}
\begin{aligned}
     &\sqrt{n}\left(\widehat \mu_{\mathrm{AIPW}}-\mu_n \right) = \frac{1}{\sqrt{n}} \sum_{i=1}^n \widetilde W_{i} - M_1\sqrt{n}(\widehat \eta - \eta) + M_2\sqrt{n}(\widehat \theta - \theta) + o_p(1) ~.
\end{aligned}
\end{align}
    
   Consider the setting where the outcome and propensity score model estimators are obtained via GMM first. The  moment conditions are $\frac{1}{n}\sum_{i=1}^n g^y(Y_i, X_i^y;\widehat\theta)=0$ and $\frac{1}{n}\sum_{i=1}^n g^a(Y_i, X_i^a;\widehat\eta)=0$. As shown in the proof of  Proposition \ref{prop:gmm-consistency}, we have
    \begin{align}\label{eqn:clt_eqn5}
        \sqrt{n}(\widehat \theta - \theta) = - M_3^{-1}  \frac{1}{\sqrt{n}}\sum_{i=1}^n g^y(Y_i, X_i^y;\theta) + o_p(1) ~.
    \end{align}
    and
    \begin{align}\label{eqn:clt_eqn6}
        \sqrt{n}(\widehat \eta - \eta) = -M_4^{-1} \frac{1}{\sqrt{n}}\sum_{i=1}^n g^a(Y_i, X_i^a;\theta) + o_p(1) ~,
    \end{align}
    where $M_3$ and $M_4$ are defined as in \eqref{eqn:m3} and  \eqref{eqn:m4}.
    
    Combing \eqref{eqn:clt_eqn3}, \eqref{eqn:clt_eqn5}, and \eqref{eqn:clt_eqn6}, we have
    \begin{align*}
        &\sqrt{n}\left(\widehat \mu_{\mathrm{AIPW}}-\mu_n \right) \\
        &= \frac{1}{\sqrt{n}} \sum_{i=1}^n \Bigg[\frac{1\{ \mathbf{A}_{\cN_i} = \mathbf{a}_{\cN_i}\} }{\pi_{\mathbf{a}_{\cN_i}}(X_i^a;\eta)} \left(Y_i-\beta_{\mathbf{a}_{\cN_i}}(X_i^y;\theta)\right) +\beta_{\mathbf{a}_{\cN_i}}(X_i^y;\theta) -  \mu_i \\
        &\quad - M_1 \mathrm{IF}_{ \eta}(\mathbf{A}_{\cN_i},X_i^a) + M_2 \mathrm{IF}_{\theta}(Y_i, X_i^y)  \Bigg]+ o_p(1) = \frac{1}{\sqrt{n}}\sum_{i=1}^n  W_i + o_p(1) ~,
    \end{align*}
    where
    \begin{align*}
    \mathrm{IF}_{\theta}(Y_i, X_i^y) &= -M_3^{-1} g^y(Y_i, X_i^y;\theta)\\
    \mathrm{IF}_{\eta}(\mathbf{A}_{\cN_i},X_i^a) &= -M_4^{-1}  g^a(\mathbf{A}_{\cN_i}, X_i^a;\eta)~.
    \end{align*}
 Assumptions~\ref{ass:network_weak_dependence} and~\ref{ass:regularization} guarantee the network weak dependence of $W_i$. Following \citet{kojevnikov2021limit}, we then obtain the desired CLT in Theorem \ref{thm:clt}. In the general case where $\widehat\theta$ and $\widehat\eta$ are estimated using methods other than GMM, $\mathrm{IF}_{\theta}$ and $\mathrm{IF}_{\eta}$ should be replaced by their corresponding influence functions.
\end{proof}
\subsection{Proof of Theorem \ref{thm:exposure-1}}
\begin{proof}
As in the proof of Theorem~\ref{thm:exposure-1}, we decompose the left-hand side of \eqref{eqn:exposure-1} into three parts and analyze each separately. Sepcifically, we write 
 \begin{align}\label{eqn:proof-exposure-1}
 \begin{aligned}
        \text{LHS of \eqref{eqn:exposure-1}}=  &\mE\left[Y_i \mid \mathbf{Y}_{-i}, \mathbf{A}_{\cN_i}=\mathbf{a}_{\cN_i}, \mathbf{A}_{-\cN_i}, \mathbf{L}]\right] +\mE\left[  \frac{1\left\{ T_i = t \right\} Y_i}{\mP\left( T_i = t\mid \mathbf{L}, \mathbf{A}_{-\cN_i} \right)}\right] \\&- \mE\left[  \frac{1\left\{ T_i = t \right\} \mE[Y_i \mid \mathbf{Y}_{-i}, \mathbf{A}_{\cN_i}=\mathbf{a}_{\cN_i}, \mathbf{A}_{-\cN_i}, \mathbf{L}]}{\mP\left( T_i = t\mid \mathbf{L}, \mathbf{A}_{-\cN_i} \right)}\right]\\
       =& \breve \cD_1 + \breve \cD_2 - \breve \cD_3~.    
 \end{aligned}
    \end{align}
By using \eqref{eqn:thm1_1}, we have $\breve \cD_1 = \cD_1 = \mE[Y_i(\mathbf{a}_{\cN_i}, \mathbf{A}_{-\cN_i})]$ immediately. For the second term, it satisfies    
\begin{align*}
    \begin{aligned}
        \breve\cD_2=&\mE\left[  \frac{ \mE\left[1\{ T_i = t \} Y_i \mid \mathbf{L},\mathbf{A}_{-\cN_i}\right]}{\mP(T_i = t \mid \mathbf{L}, \mathbf{A}_{-\cN_i} )}  \right]\\
        =&\mE\left[  \frac{ \mE[1\{ T_i = t \} \sum_{\mathbf{a}} Y_i(\mathbf{a}) 1\{\mathbf{A}=\mathbf{a}\}\mid\mathbf{L},\mathbf{A}_{-\cN_i}]}{\mP(T_i=t \mid \mathbf{L}, \mathbf{A}_{-\cN_i} )}  \right] \\
        =&\mE\left[\sum_{\mathbf{a}_{-\cN_i}} \sum_{\mathbf{a}_{\cN_i}} 1\{\mathbf{A}_{-\cN_i}=\mathbf{a}_{-\cN_i}\}  \frac{ \mE[1\{ T_i = t \}  Y_i(\mathbf{a}) 1\{\mathbf{A}_{\cN_i}=\mathbf{a}_{\cN_i}\}\mid\mathbf{L},\mathbf{A}_{-\cN_i}]}{\mP(T_i=t \mid \mathbf{L}, \mathbf{A}_{-\cN_i} )}  \right]\\
        =&\mE\left[\sum_{\mathbf{a}_{-\cN_i}} \sum_{\mathbf{a}_{\cN_i}} 1\{\mathbf{A}_{-\cN_i}=\mathbf{a}_{-\cN_i}\}  \frac{ \mP(T_i = t,\mathbf{A}_{\cN_i}=\mathbf{a}_{\cN_i})\mid\mathbf{L},\mathbf{A}_{-\cN_i})}{\mP(T_i=t \mid \mathbf{L}, \mathbf{A}_{-\cN_i} )} \mE[ Y_i(\mathbf{a})\mid\mathbf{L},\mathbf{A}_{-\cN_i}] \right] \\
        =&\mE\left[\sum_{\mathbf{a}_{\cN_i}}\mP(\mathbf{A}_{\cN_i}=\mathbf{a}_{\cN_i}\mid T_i = t,\mathbf{L},\mathbf{A}_{-\cN_i}) \sum_{\mathbf{a}_{-\cN_i}}  \mE[  1\{\mathbf{A}_{-\cN_i}=\mathbf{a}_{-\cN_i}\} Y_i(\mathbf{a})\mid\mathbf{L},\mathbf{A}_{-\cN_i}] \right]\\
        =& \mE\left[\sum_{\mathbf{a}_{\cN_i}}  \mP(\mathbf{A}_{\cN_i}=\mathbf{a}_{\cN_i}\mid T_i = t,\mathbf{L},\mathbf{A}_{-\cN_i}) \sum_{\mathbf{a}_{-\cN_i}}   1\{\mathbf{A}_{-\cN_i}=\mathbf{a}_{-\cN_i}\} Y_i(\mathbf{a}_{\cN_i}, \mathbf{a}_{-\cN_i}) \right] \\
        =& \mE\left[\sum_{\mathbf{a}_{\cN_i}} \mP(\mathbf{A}_{\cN_i}=\mathbf{a}_{\cN_i}\mid T_i = t,\mathbf{L},\mathbf{A}_{-\cN_i})  Y_i(\mathbf{a}_{\cN_i}, \mathbf{A}_{-\cN_i}) \right]  ~.
    \end{aligned}
\end{align*}
The third term can be written as 
    \begin{align*}
    \begin{aligned}
       \breve \cD_3=& \mE\left[  \frac{1\{ T_i = t \} \mE[Y_i(\mathbf{a}_{\cN_i}, \mathbf{A}_{-\cN_i}) \mid \mathbf{Y}_{-i}(\mathbf{a}_{\cN_i}, \mathbf{A}_{-\cN_i}), \mathbf{A}_{-\cN_i},\mathbf{L}]}{\mP(T_i = t \mid \mathbf{L}, \mathbf{A}_{-\cN_i} )} \right]\\
        =& \mE\left[\mE\left[  \frac{1\{ T_i = t \} \mE[Y_i(\mathbf{a}_{\cN_i}, \mathbf{A}_{-\cN_i}) \mid \mathbf{Y}_{-i}(\mathbf{a}_{\cN_i}, \mathbf{A}_{-\cN_i}), \mathbf{A}_{-\cN_i},\mathbf{L}]}{\mP(T_i = t \mid \mathbf{L}, \mathbf{A}_{-\cN_i} )}   \mid \mathbf{L},\mathbf{A} \right] \right]\\
        =&\mE\left[  \frac{1\{ T_i = t \} \mE[\mE[Y_i(\mathbf{a}_{\cN_i}, \mathbf{A}_{-\cN_i}) \mid \mathbf{Y}_{-i}(\mathbf{a}_{\cN_i}, \mathbf{A}_{-\cN_i}), \mathbf{A}_{-\cN_i},\mathbf{L}]\mid\mathbf{L},\mathbf{A}]}{\mP(T_i = t \mid \mathbf{L}, \mathbf{A}_{-\cN_i} )}  \right] \\
        =&\mE\left[  \frac{1\{ T_i = t \} \mE[\mE[Y_i(\mathbf{a}_{\cN_i}, \mathbf{A}_{-\cN_i}) \mid \mathbf{Y}_{-i}(\mathbf{a}_{\cN_i}, \mathbf{A}_{-\cN_i}), \mathbf{A}_{-\cN_i},\mathbf{L}]\mid\mathbf{A}_{-\cN_i},\mathbf{L}]}{\mP(T_i = t \mid \mathbf{L}, \mathbf{A}_{-\cN_i} )}  \right]\\
        =& \mE\left[\mE\left[  \frac{1\{ T_i = t \} \mE[Y_i(\mathbf{a}_{\cN_i}, \mathbf{A}_{-\cN_i}) \mid \mathbf{A}_{-\cN_i},\mathbf{L}]}{\mP(T_i = t \mid \mathbf{L}, \mathbf{A}_{-\cN_i} )}  \mid \mathbf{L}, \mathbf{A}_{-\cN_i} \right]\right]\\
        =&\mE\left[\frac{\mE[Y_i(\mathbf{a}_{\cN_i}, \mathbf{A}_{-\cN_i}) \mid  \mathbf{A}_{-\cN_i},\mathbf{L}]}{\mP(T_i = t \mid \mathbf{L}, \mathbf{A}_{-\cN_i} ) }\mE\left[  1\{ T_i = t \}  \mid \mathbf{L}, \mathbf{A}_{-\cN_i} \right]\right]\\
        =& \mE[\mE[Y_i(\mathbf{a}_{\cN_i}, \mathbf{A}_{-\cN_i}) \mid \mathbf{A}_{-\cN_i},\mathbf{L}]]  = \mE[Y_i(\mathbf{a}_{\cN_i}, \mathbf{A}_{-\cN_i})]~.
         \end{aligned}
    \end{align*}
Substituting the results of $\breve\cD_1$, $\breve\cD_2$, and $\breve\cD_3$ into \eqref{eqn:proof-exposure-1} completes the proof of Theorem~\ref{thm:exposure-1}.

\end{proof}
\subsection{Proof of Theorem \ref{thm:exposure-2}}
\begin{proof}
In light of Theorem~\ref{thm:exposure-1}, the proof of Theorem~\ref{thm:exposure-2} reduces to verifying that
    \begin{align*}
        \mE\left[\sum_{\mathbf{a}_{\cN_i}} \mP(\mathbf{A}_{\cN_i}=\mathbf{a}_{\cN_i}\mid T_i = t,\mathbf{L},\mathbf{A}_{-\cN_i})  Y_i(\mathbf{a}_{\cN_i}, \mathbf{A}_{-\cN_i}) \right] = \mE[Y_i(\mathbf{a})]~.
    \end{align*}
We consider two separate cases and show that this equality holds under either Assumption~\ref{ass:exposure}(i) or (ii), confirming the double robustness property.
\subsubsection*{Case 1: Assumption~\ref{ass:exposure}(i) holds.}
It is straightforward to see the equality holds under Assumption \ref{ass:exposure}(i).
\subsubsection*{Case 2: Assumption~\ref{ass:exposure}(ii) holds.}
    Note that
    \begin{align*}
         &\mE\left[\sum_{\mathbf{a}_{\cN_i}} \mP(\mathbf{A}_{\cN_i}=\mathbf{a}_{\cN_i}\mid T_i = t,\mathbf{L},\mathbf{A}_{-\cN_i})  Y_i(\mathbf{a}_{\cN_i}, \mathbf{A}_{-\cN_i}) \right] \\
         =& \mE\left[\sum_{\mathbf{a}_{\cN_i}} \mP(\mathbf{A}_{\cN_i}=\mathbf{a}_{\cN_i}\mid T_i = t,\mathbf{L},\mathbf{A}_{-\cN_i})  \mE\left[Y_i(\mathbf{a}_{\cN_i}, \mathbf{A}_{-\cN_i}) \mid T_i = t,\mathbf{L},\mathbf{A}_{-\cN_i} \right] \right] \\
         =&\mE\left[\sum_{\mathbf{a}_{\cN_i}}   \mE\left[1\{\mathbf{A}_{\cN_i} = \mathbf{a}_{\cN_i}\} Y_i(\mathbf{a}_{\cN_i}, \mathbf{A}_{-\cN_i}) \mid T_i = t,\mathbf{L},\mathbf{A}_{-\cN_i} \right] \right] \\
         =& \mE\left[ \mE\left[ Y_i\mid T_i = t,\mathbf{L},\mathbf{A}_{-\cN_i} \right] \right] = \mE\left[ \mE\left[ Y_i\mid \mathbf{Y}_{-i}, T_i = t,\mathbf{L} \right] \right] \\
         =& \mE\left[ \mE[Y_i \mid \mathbf{Y}_{-i}, \mathbf{A}=\mathbf{a}, \mathbf{L}]\right] = \mE[Y_i(\mathbf{a})] ~. 
    \end{align*}
    The second equality holds under
    \begin{equation}\label{eq:app-cond-indep-t}
        \mathbf{Y}(\mathbf{a}_{\cN_i}, \mathbf{a}_{-\cN_i}) \independent \mathbf{A}_{\cN_i} \mid T_i,  \mathbf{A}_{-\cN_i}, \mathbf{L} ~.
    \end{equation}
    And the second last equality holds under Assumption \ref{ass:exposure}(ii). 
\end{proof}

\subsection{Proof of Theorem \ref{thm:identification_extra}}
\begin{proof}
    From Lemma \ref{lemma:independent}(i), we have
    \begin{equation}\label{eq:app-cond-indep2}
        Y_i(\mathbf{a}_{\cN_i}, \mathbf{a}_{-\cN_i}) \independent \mathbf{A}_{\cN_i} \mid \mathbf{A}_{-\cN_i}, \mathbf{L} ~,
    \end{equation}
    Consequently, 
    \begin{align*}
        \mE\left[ \mE[Y_i \mid \mathbf{A}_{\cN_i}=\mathbf{a}_{\cN_i}, \mathbf{A}_{-\cN_i}, \mathbf{L}] \right] &= \mE\left[ \mE[Y_i(\mathbf{a}_{\cN_i}, \mathbf{A}_{-\cN_i}) \mid \mathbf{A}_{\cN_i}=\mathbf{a}_{\cN_i}, \mathbf{A}_{-\cN_i}, \mathbf{L}] \right]  \\
        &= \mE\left[ \mE[Y_i(\mathbf{a}_{\cN_i}, \mathbf{A}_{-\cN_i}) \mid \mathbf{A}_{-\cN_i}, \mathbf{L}] \right] \\
        &= \mE[Y_i(\mathbf{a}_{\cN_i}, \mathbf{A}_{-\cN_i})] ~.
    \end{align*}
    Then,
    \begin{align*}
        \mE\left[  \frac{1\left\{ \mathbf{A}_{\cN_i} = \mathbf{a}_{\cN_i} \right\} Y_i}{\mP\left( \mathbf{A}_{\cN_i} = \mathbf{a}_{\cN_i} \mid \mathbf{L}, \mathbf{A}_{-\cN_i} \right)}  \right] &= \mE\left[ \mE\left[  \frac{1\left\{ \mathbf{A}_{\cN_i} = \mathbf{a}_{\cN_i} \right\} Y_i}{\mP\left( \mathbf{A}_{\cN_i} = \mathbf{a}_{\cN_i} \mid \mathbf{L}, \mathbf{A}_{-\cN_i} \right)} \mid \mathbf{L}, \mathbf{A}_{-\cN_i} \right] \right] \\
        &= \mE\left[   \frac{\mE\left[1\left\{ \mathbf{A}_{\cN_i} = \mathbf{a}_{\cN_i} \right\} Y_i(\mathbf{a}_{\cN_i},  \mathbf{A}_{-\cN_i} )\mid \mathbf{L}, \mathbf{A}_{-\cN_i} \right]  }{\mP\left( \mathbf{A}_{\cN_i} = \mathbf{a}_{\cN_i} \mid \mathbf{L}, \mathbf{A}_{-\cN_i} \right)}  \right] \\
        &=  \mE\left[  \mE\left[ Y_i(\mathbf{a}_{\cN_i},  \mathbf{A}_{-\cN_i} )\mid \mathbf{L}, \mathbf{A}_{-\cN_i} \right]   \right] \hspace{3mm} \text{(by Eqn. (\ref{eq:app-cond-indep2}))} \\
        &= \mE\left[ Y_i(\mathbf{a}_{\cN_i}, \mathbf{A}_{-\cN_i}) \right] ~.
    \end{align*}
    Finally,
    \begin{align*}
        \mE\left[  \frac{1\left\{ \mathbf{A}_{\cN_i} = \mathbf{a}_{\cN_i} \right\}  \mE[Y_i \mid \mathbf{A}_{\cN_i}=\mathbf{a}_{\cN_i}, \mathbf{A}_{-\cN_i}, \mathbf{L}]}{\mP\left( \mathbf{A}_{\cN_i} = \mathbf{a}_{\cN_i} \mid \mathbf{L}, \mathbf{A}_{-\cN_i} \right)}\right] &= \mE\left[ \mE[Y_i \mid \mathbf{A}_{\cN_i}=\mathbf{a}_{\cN_i}, \mathbf{A}_{-\cN_i}, \mathbf{L}] \right]\\
        &= \mE\left[ Y_i(\mathbf{a}_{\cN_i}, \mathbf{A}_{-\cN_i}) \right] ~.
    \end{align*}
    Therefore, the result holds. 
\end{proof}
\subsection{Technical Lemmas and Their Proofs}
\begin{lemma}\label{lemma:independent}
Under Assumptions \ref{ass:consistency}, \ref{ass:ignorability} and \ref{ass:positivity}, we have
\begin{enumerate}[label={(\roman*)}]
    \item  $\mathbf{Y}(\mathbf{a}_{\cN_i}, \mathbf{A}_{-\cN_i})\independent \mathbf{A}_{\cN_i}|\mathbf{A}_{-\cN_i}, \mathbf{L}$~.
    \item  The distributions of $Y_i(\mathbf{a}_{\cN_i}, \mathbf{A}_{-\cN_i}) \mid \mathbf{Y}_{-i}, \mathbf{A}_{\cN_i}=\mathbf{a}_{\cN_i}, \mathbf{A}_{-\cN_i}, \mathbf{L}$ and $Y_i(\mathbf{a}_{\cN_i}, \mathbf{A}_{-\cN_i}) \mid \mathbf{Y}_{-i}(\mathbf{a}_{\cN_i}, \mathbf{A}_{-\cN_i}),\mathbf{A}_{-\cN_i}, \mathbf{L}$ are identical.
    \item  The distributions of $Y_i \mid \mathbf{Y}_{-i}, \mathbf{A}=\mathbf{a}, \mathbf{L}$ and $Y_i(\mathbf{a}) \mid \mathbf{Y}_{-i}(\mathbf{a}), \mathbf{L}$ are identical.
\end{enumerate}
\end{lemma}
\begin{proof}

(1) For any Borel set $B$, $\mathbf{a}^*_{\cN_i}$ and $\mathbf{a}^*_{-\cN_i}$ satisfying $\mP(\mathbf{A}_{-\cN_i}=\mathbf{a}^*_{-\cN_i} \mid  \mathbf{L})>0$, we have 
\begin{align*}
    &\mP(\mathbf{Y}(\mathbf{a}_{\cN_i}, \mathbf{A}_{-\cN_i})\in B, \mathbf{A}_{\cN_i}=\mathbf{a}^*_{\cN_i} \mid \mathbf{A}_{-\cN_i}=\mathbf{a}^*_{-\cN_i}, \mathbf{L}) \\
    &= \frac{\mP(\mathbf{Y}(\mathbf{a}_{\cN_i}, \mathbf{a}^*_{-\cN_i})\in B, \mathbf{A}_{\cN_i}= \mathbf{a}^*_{\cN_i}, \mathbf{A}_{-\cN_i}=\mathbf{a}^*_{-\cN_i} \mid  \mathbf{L})}{\mP(\mathbf{A}_{-\cN_i} =\mathbf{a}^*_{-\cN_i}\mid  \mathbf{L})} \\
    &=\frac{\mP(\mathbf{Y}(\mathbf{a}_{\cN_i}, \mathbf{a}^*_{-\cN_i})\in B \mid  \mathbf{L})\mP( \mathbf{A}_{\cN_i}=\mathbf{a}^*_{\cN_i}, \mathbf{A}_{-\cN_i}=\mathbf{a}^*_{-\cN_i} \mid  \mathbf{L})}{\mP(\mathbf{A}_{-\cN_i}=\mathbf{a}^*_{-\cN_i} \mid  \mathbf{L})} \\
    &= \frac{\mP(\mathbf{Y}(\mathbf{a}_{\cN_i}, \mathbf{a}^*_{-\cN_i})\in B \mid  \mathbf{A}_{-\cN_i}=\mathbf{a}^*_{-\cN_i}, \mathbf{L})\mP( \mathbf{A}_{\cN_i}=\mathbf{a}^*_{\cN_i}, \mathbf{A}_{-\cN_i}=\mathbf{a}^*_{-\cN_i} \mid  \mathbf{L})}{\mP(\mathbf{A}_{-\cN_i} =\mathbf{a}^*_{-\cN_i} \mid  \mathbf{L})} \\
    &= \mP(\mathbf{Y}(\mathbf{a}_{\cN_i}, \mathbf{A}_{-\cN_i})\in B\mid  \mathbf{A}_{-\cN_i}=\mathbf{a}^*_{-\cN_i}, \mathbf{L}) \mP( \mathbf{A}_{\cN_i}=\mathbf{a}^*_{\cN_i} \mid  \mathbf{A}_{-\cN_i}=\mathbf{a}^*_{-\cN_i}, \mathbf{L}) ~.
\end{align*}
This completes the proof of $\mathbf{Y}(\mathbf{a}_{\cN_i}, \mathbf{A}_{-\cN_i}) \independent \mathbf{A}_{\cN_i} \mid \mathbf{A}_{-\cN_i}, \mathbf{L}$.

(2) For any Borel sets $B_1, B_2$ satisfying $\mP(\mathbf{Y}_{-i}\in B_2,  \mathbf{A}_{\cN_i}=\mathbf{a}_{\cN_i}\mid\mathbf{A}_{-\cN_i}, \mathbf{L}) > 0$, we have
\begin{align}
\begin{aligned}
&\mP(Y_i(\mathbf{a}_{\cN_i}, \mathbf{A}_{-\cN_i})\in B_1 \mid \mathbf{Y}_{-i}\in B_2,  \mathbf{A}_{\cN_i}=\mathbf{a}_{\cN_i}, \mathbf{A}_{-\cN_i}, \mathbf{L}) \\
& = \frac{\mP(Y_i(\mathbf{a}_{\cN_i}, \mathbf{A}_{-\cN_i})\in B_1, \mathbf{Y}_{-i}\in B_2 \mid \mathbf{A}_{\cN_i}=\mathbf{a}_{\cN_i}, \mathbf{A}_{-\cN_i}, \mathbf{L})}{\mP(\mathbf{Y}_{-i}\in B_2 \mid \mathbf{A}_{\cN_i}=\mathbf{a}_{\cN_i}, \mathbf{A}_{-\cN_i}, \mathbf{L})} \\
& = \frac{\mP(Y_i(\mathbf{a}_{\cN_i}, \mathbf{A}_{-\cN_i})\in B_1, \mathbf{Y}_{-i}(\mathbf{a}_{\cN_i}, \mathbf{A}_{-\cN_i})\in B_2 \mid \mathbf{A}_{\cN_i}=\mathbf{a}_{\cN_i}, \mathbf{A}_{-\cN_i}, \mathbf{L})}{\mP(\mathbf{Y}_{-i}(\mathbf{a}_{\cN_i}, \mathbf{A}_{-\cN_i})\in B_2 \mid \mathbf{A}_{\cN_i}=\mathbf{a}_{\cN_i}, \mathbf{A}_{-\cN_i}, \mathbf{L})} \\
& = \frac{\mP(Y_i(\mathbf{a}_{\cN_i}, \mathbf{A}_{-\cN_i})\in B_1, \mathbf{Y}_{-i}(\mathbf{a}_{\cN_i}, \mathbf{A}_{-\cN_i})\in B_2 \mid  \mathbf{A}_{-\cN_i}, \mathbf{L})}{\mP(\mathbf{Y}_{-i}(\mathbf{a}_{\cN_i}, \mathbf{A}_{-\cN_i})\in B_2 \mid \mathbf{A}_{-\cN_i}, \mathbf{L})} \\
& = \mP(Y_i(\mathbf{a}_{\cN_i}, \mathbf{A}_{-\cN_i})\in B_1 \mid \mathbf{Y}_{-i}(\mathbf{a}_{\cN_i}, \mathbf{A}_{-\cN_i})\in B_2, \mathbf{A}_{-\cN_i}, \mathbf{L}),
\end{aligned}
\end{align}
where the first equality follows from the definition of conditional probability, and the second equation follows from Assumption~\ref{ass:consistency}. The third equality follows from the result in (1).

(3) Similar to the proof of (2), for any Borel sets $B_1, B_2$ satisfying $\mP(\mathbf{Y}_{-i}\in B_2\mid\mathbf{A}=\mathbf{a}, \mathbf{L}) > 0$, we have
\begin{align}\label{eqn:thm1_2_2}
\begin{aligned}
&\mP(Y_i(\mathbf{a})\in B_1 \mid \mathbf{Y}_{-i}\in B_2, \mathbf{A}=\mathbf{a}, \mathbf{L}) \\
& = \frac{\mP(Y_i(\mathbf{a})\in B_1, \mathbf{Y}_{-i}\in B_2 \mid \mathbf{A}=\mathbf{a}, \mathbf{L})}{\mP(\mathbf{Y}_{-i}\in B_2 \mid \mathbf{A}=\mathbf{a}, \mathbf{L})} \\
& = \frac{\mP(Y_i(\mathbf{a})\in B_1, \mathbf{Y}_{-i}(\mathbf{a})\in B_2 \mid \mathbf{A}=\mathbf{a}, \mathbf{L})}{\mP(\mathbf{Y}_{-i}(\mathbf{a})\in B_2 \mid \mathbf{A}=\mathbf{a}, \mathbf{L})} \\
& = \frac{\mP(Y_i(\mathbf{a})\in B_1, \mathbf{Y}_{-i}(\mathbf{a})\in B_2 \mid \mathbf{L})}{\mP(\mathbf{Y}_{-i}(\mathbf{a})\in B_2 \mid \mathbf{L})} \\
& = \mP(Y_i(\mathbf{a})\in B_1 \mid \mathbf{Y}_{-i}(\mathbf{a})\in B_2, \mathbf{L})~.
\end{aligned}
\end{align}
This completes the proof of (3).
\end{proof}

\begin{lemma}\label{lemma:M-consistency}
    Under Assumption \ref{ass:network_weak_dependence} and \ref{ass:regularization}, we have
  \begin{align*}
      \frac{1}{n}\sum_{i=1}^n H_{1,i} \overset{a.s.}{\longrightarrow} M_1,\qquad \frac{1}{n}\sum_{i=1}^n H_{2,i} \overset{a.s.}{\longrightarrow} M_2~
  \end{align*}
  and 
   \begin{align*}
     \frac{1}{n}\sum_{i=1}^n\frac{\partial g^y(Y_i,X_i^y;\theta)}{\partial\theta} \overset{a.s.}{\longrightarrow} \lim_{n\to\infty} M_3, \qquad
     \frac{1}{n}\sum_{i=1}^n\frac{\partial g^a(Y_i,X_i^a;\eta)}{\partial\eta} \overset{a.s.}{\longrightarrow} \lim_{n\to\infty} M_4~.
  \end{align*}
\end{lemma}
\begin{proof}
It is sufficient to verify the first equation. The proofs for the remaining three are analogous. Note that $H_{1,i} = H_{1,i}(Z_i)$ can be viewed as a function of $Z_i := (Y_j,A_j,L_j,j\in\cN_n(i,K+1))$, and from the definition of $\psi$ dependence, we directly have $Z_i$ is $\psi$ dependence with coefficients $\widetilde\zeta_n = \zeta_{n,s-2k-2}$ for $s>2k+2$ and $\widetilde \zeta $ holds the same asymptotic assumptions as in Assumption \ref{ass:network_weak_dependence}. Then for any Lipschitz function $h$ and any bounded Lipschitz function $g$,  as the composition $g\circ h$ is also a bounded Lipschitz function, we have $h(Z_i)$ is also $\psi$ dependence with the same coefficients immediately from the definition of $\psi$ dependence. Consequently, theorems in \cite{kojevnikov2021limit} ensure the CLT result of $h(Z_i)$ for any Lipschitz function $h$. 
    
    However, $H_{1,i}$ is not Lipschitz in general. Therefore, we introduce the truncated function 
    $$H_{1,i}^{(n)} = H_{1,i}^{(n)}(Z_i) = \begin{cases}
        H_{1,i}(Z_i) ,&\quad  |H_{1,i}(Z_i)|\le n\\
        n,&\quad H_{1,i}(Z_i)> n\\
         -n,&\quad H_{1,i}(Z_i)<-n
    \end{cases}$$
    and it is sufficient to show that
    \begin{align}
         &\frac{1}{n}\sum_{i=1}^n(H_{1,i}- H_{1,i}^{(n)}) \overset{a.s.}{\longrightarrow} 0~, \label{thm4.3:eqn1}
    \end{align}
    and
    \begin{align}
         &\frac{1}{n}\sum_{i=1}^n H_{1,i}^{(n)} \overset{a.s.}{\longrightarrow} M_1 \label{thm4.3:eqn2}~.
    \end{align}
    Note that, for any $\epsilon>0$, we have
    \begin{align*}
        \mathrm{P}\left(\left|\frac{1}{n}\sum_{i=1}^n(H_{1,i}- H_{1,i}^{(n)})\right|\ge \epsilon\right) \le \sum_{i=1}^n \mathrm{P}(H_{1,i}\neq H_{1,i}^{(n)}) =\sum_{i=1}^n \mathrm{P}(|H_{1,i}|> n)\le \sum_{i=1}^n\frac{\mE(H_{1,i}^4)}{n^4}\lesssim \frac{1}{n^{3}}~,
    \end{align*}
    which implies that
    \begin{align}
        \sum_{n=1}^\infty\mathrm{P}\left(\left|\frac{1}{n}\sum_{i=1}^n(H_{1,i}- H_{1,i}^{(n)})\right|\ge \epsilon\right) \lesssim \sum_{n=1}^\infty\frac{1}{n^{3}}< \infty~ \label{thm4.3:eqn3}.
    \end{align}
    Consequently, \eqref{thm4.3:eqn1} is implied by \eqref{thm4.3:eqn3} and Borel-Cantelli Lemma. To show \eqref{thm4.3:eqn2}, recall that $H_{1,i}^{(n)}$ is a Lipshcitz function with Lipshcitz coefficient $n^q$, then for any bounded Lipschitz function $f_1$, $f_2$, we have
    \begin{align*}
    |\Cov(f_1(H_{1,Q_1}), f_2(H_{1,Q_2}) \mid \cG_n)| \leq \psi_{q_1,q_2}(f_1\circ h_{1,Q_1},f_2\circ h_{1,Q_2}) \zeta_{n,s} n^{2q}  \ \ \  \mbox{a.s.}
    \end{align*}
    Therefore, $H_{1,i}$ is $\psi-$ dependence with coefficient $\widetilde\zeta_{n,s} = \zeta_{n,s} n^{2q}$ and $\psi_{q_1,q_2}$, following \cite{kojevnikov2021limit}, we have 
    \begin{align*}
        \frac{1}{n}\sum_{i=1}^n (H_{1,i}^{(n)}-\mathrm E(H_{1,i}^{(n)})) \overset{a.s.}{\longrightarrow} 0.
    \end{align*}
    This completes the proof.
\end{proof}

\newpage
\bibliography{reference.bib}

\begin{thebibliography}{52}
\expandafter\ifx\csname natexlab\endcsname\relax\def\natexlab#1{#1}\fi
\expandafter\ifx\csname url\endcsname\relax
  \def\url#1{\texttt{#1}}\fi
\expandafter\ifx\csname urlprefix\endcsname\relax\def\urlprefix{URL }\fi
\providecommand{\eprint}[2][]{\url{#2}}

\bibitem[{Aronow(2012)}]{Aronow2012}
\textsc{Aronow, P.~M.} (2012).
\newblock A general method for detecting interference between units in randomized experiments.
\newblock \textit{Sociological Methods \& Research}, \textbf{41} 3--16.

\bibitem[{Aronow and Samii(2017)}]{Aronow2017}
\textsc{Aronow, P.~M.} and \textsc{Samii, C.} (2017).
\newblock Estimating average causal effects under general interference, with application to a social network experiment.
\newblock \textit{The Annals of Applied Statistics}, \textbf{11} 1912--1947.

\bibitem[{Athey et~al.(2018)Athey, Eckles and Imbens}]{athey2018exact}
\textsc{Athey, S.}, \textsc{Eckles, D.} and \textsc{Imbens, G.~W.} (2018).
\newblock Exact p-values for network interference.
\newblock \textit{Journal of the American Statistical Association}, \textbf{113} 230--240.

\bibitem[{Barkley et~al.(2020)Barkley, Hudgens, Clemens, Ali, Emch et~al.}]{barkley2020causal}
\textsc{Barkley, B.~G.}, \textsc{Hudgens, M.~G.}, \textsc{Clemens, J.~D.}, \textsc{Ali, M.}, \textsc{Emch, M.~E.} \textsc{et~al.} (2020).
\newblock Causal inference from observational studies with clustered interference, with application to a cholera vaccine study.
\newblock \textit{Annals of Applied Statistics}, \textbf{14} 1432--1448.

\bibitem[{Besag(1974)}]{Besag1974}
\textsc{Besag, J.} (1974).
\newblock Spatial interaction and the statistical analysis of lattice systems.
\newblock \textit{Journal of the Royal Statistical Society. Series B (Methodological)}, \textbf{36} 192--236.

\bibitem[{Bhattacharya et~al.(2020)Bhattacharya, Malinsky and Shpitser}]{Bhattacharya2020}
\textsc{Bhattacharya, R.}, \textsc{Malinsky, D.} and \textsc{Shpitser, I.} (2020).
\newblock Causal inference under interference and network uncertainty.
\newblock In \textit{Uncertainty in Artificial Intelligence}. PMLR, 1028--1038.

\bibitem[{Bhattacharya and Sen(2024)}]{bhattacharya2024causal}
\textsc{Bhattacharya, S.} and \textsc{Sen, S.} (2024).
\newblock Causal effect estimation under network interference with mean-field methods.
\newblock \textit{arXiv preprint arXiv:2407.19613}.

\bibitem[{Bowers et~al.(2013)Bowers, Fredrickson and Panagopoulos}]{Bowers2013}
\textsc{Bowers, J.}, \textsc{Fredrickson, M.~M.} and \textsc{Panagopoulos, C.} (2013).
\newblock Reasoning about interference between units: A general framework.
\newblock \textit{Political Analysis}, \textbf{21} 97--124.

\bibitem[{Bramoullé et~al.(2009)Bramoullé, Djebbari and Fortin}]{BRAMOULLE2009}
\textsc{Bramoullé, Y.}, \textsc{Djebbari, H.} and \textsc{Fortin, B.} (2009).
\newblock Identification of peer effects through social networks.
\newblock \textit{Journal of Econometrics}, \textbf{150} 41--55.

\bibitem[{Cox(1958)}]{cox1958planning}
\textsc{Cox, D.~R.} (1958).
\newblock \textit{Planning of experiments.}
\newblock Wiley.

\bibitem[{Egami(2018)}]{egami2020}
\textsc{Egami, N.} (2018).
\newblock Identification of causal diffusion effects under structural stationarity.
\newblock \textit{arXiv preprint arXiv:1810.07858}.

\bibitem[{Egami and Tchetgen~Tchetgen(2023)}]{Egami2023}
\textsc{Egami, N.} and \textsc{Tchetgen~Tchetgen, E.~J.} (2023).
\newblock Identification and estimation of causal peer effects using double negative controls for unmeasured network confounding.
\newblock \textit{Journal of the Royal Statistical Society Series B: Statistical Methodology}, \textbf{86} 487--511.

\bibitem[{Ferracci et~al.(2014)Ferracci, Jolivet and van~den Berg}]{Ferracci2014}
\textsc{Ferracci, M.}, \textsc{Jolivet, G.} and \textsc{van~den Berg, G.~J.} (2014).
\newblock Evidence of treatment spillovers within markets.
\newblock \textit{The Review of Economics and Statistics}, \textbf{96} 812--823.

\bibitem[{Forastiere et~al.(2021)Forastiere, Airoldi and Mealli}]{forastiere2021identification}
\textsc{Forastiere, L.}, \textsc{Airoldi, E.~M.} and \textsc{Mealli, F.} (2021).
\newblock Identification and estimation of treatment and interference effects in observational studies on networks.
\newblock \textit{Journal of the American Statistical Association}, \textbf{116} 901--918.

\bibitem[{Friedman et~al.(2008)Friedman, Bolyard, Sandoval, Mateu-Gelabert, Maslow and Zenilman}]{Friedman17}
\textsc{Friedman, S.~R.}, \textsc{Bolyard, M.}, \textsc{Sandoval, M.}, \textsc{Mateu-Gelabert, P.}, \textsc{Maslow, C.} and \textsc{Zenilman, J.} (2008).
\newblock Relative prevalence of different sexually transmitted infections in hiv-discordant sexual partnerships: data from a risk network study in a high-risk new york neighbourhood.
\newblock \textit{Sexually Transmitted Infections}, \textbf{84} 17--18.

\bibitem[{Gao and Ding(2023)}]{gao2023}
\textsc{Gao, M.} and \textsc{Ding, P.} (2023).
\newblock Causal inference in network experiments: regression-based analysis and design-based properties.
\newblock \eprint{2309.07476}, \urlprefix\url{https://arxiv.org/abs/2309.07476}.

\bibitem[{Goldsmith-Pinkham and and(2013)}]{Goldsmith-Pinkham2013}
\textsc{Goldsmith-Pinkham, P.} and \textsc{and, G. W.~I.} (2013).
\newblock Social networks and the identification of peer effects.
\newblock \textit{Journal of Business \& Economic Statistics}, \textbf{31} 253--264.

\bibitem[{Graham(2008)}]{graham2008identifying}
\textsc{Graham, B.~S.} (2008).
\newblock Identifying social interactions through conditional variance restrictions.
\newblock \textit{Econometrica}, \textbf{76} 643--660.

\bibitem[{Hong and Raudenbush(2006)}]{Hong2006}
\textsc{Hong, G.} and \textsc{Raudenbush, S.~W.} (2006).
\newblock Evaluating kindergarten retention policy: A case study of causal inference for multilevel observational data.
\newblock \textit{Journal of the American Statistical Association}, \textbf{101} 901--910.

\bibitem[{Hudgens and Halloran(2008)}]{hudgens2008}
\textsc{Hudgens, M.~G.} and \textsc{Halloran, M.~E.} (2008).
\newblock Toward causal inference with interference.
\newblock \textit{Journal of the American Statistical Association}, \textbf{103} 832--842.

\bibitem[{Kojevnikov et~al.(2021)Kojevnikov, Marmer and Song}]{kojevnikov2021limit}
\textsc{Kojevnikov, D.}, \textsc{Marmer, V.} and \textsc{Song, K.} (2021).
\newblock Limit theorems for network dependent random variables.
\newblock \textit{Journal of Econometrics}, \textbf{222} 882--908.

\bibitem[{Lauritzen(1996)}]{lauritzen1996graphical}
\textsc{Lauritzen, S.} (1996).
\newblock \textit{Graphical Models}.
\newblock Oxford Statistical Science Series, Clarendon Press.

\bibitem[{Lauritzen and Richardson(2002)}]{lauritzen2002chain}
\textsc{Lauritzen, S.~L.} and \textsc{Richardson, T.~S.} (2002).
\newblock Chain graph models and their causal interpretations.
\newblock \textit{Journal of the Royal Statistical Society Series B: Statistical Methodology}, \textbf{64} 321--348.

\bibitem[{Leung(2022)}]{Leung2022}
\textsc{Leung, M.~P.} (2022).
\newblock Causal inference under approximate neighborhood interference.
\newblock \textit{Econometrica}, \textbf{90} 267--293.

\bibitem[{Liu and Hudgens(2014)}]{Liu2014}
\textsc{Liu, L.} and \textsc{Hudgens, M.~G.} (2014).
\newblock Large sample randomization inference of causal effects in the presence of interference.
\newblock \textit{Journal of the American Statistical Association}, \textbf{109} 288--301.

\bibitem[{Liu et~al.(2016)Liu, Hudgens and Becker-Dreps}]{Liu2016}
\textsc{Liu, L.}, \textsc{Hudgens, M.~G.} and \textsc{Becker-Dreps, S.} (2016).
\newblock On inverse probability-weighted estimators in the presence of interference.
\newblock \textit{Biometrika}, \textbf{103} 829--842.

\bibitem[{Liu et~al.(2019)Liu, Hudgens, Saul, Clemens, Ali and Emch}]{liu2019doubly}
\textsc{Liu, L.}, \textsc{Hudgens, M.~G.}, \textsc{Saul, B.}, \textsc{Clemens, J.~D.}, \textsc{Ali, M.} and \textsc{Emch, M.~E.} (2019).
\newblock Doubly robust estimation in observational studies with partial interference.
\newblock \textit{Stat}, \textbf{8} e214.

\bibitem[{Lundin and Karlsson(2014)}]{Lundin2014}
\textsc{Lundin, M.} and \textsc{Karlsson, M.} (2014).
\newblock Estimation of causal effects in observational studies with interference between units.
\newblock \textit{Statistical Methods \& Applications}, \textbf{23} 417--433.

\bibitem[{Manski(1993)}]{Manski1993}
\textsc{Manski, C.~F.} (1993).
\newblock Identification of endogenous social effects: The reflection problem.
\newblock \textit{The Review of Economic Studies}, \textbf{60} 531--542.

\bibitem[{Manski(2013)}]{manski2013identification}
\textsc{Manski, C.~F.} (2013).
\newblock Identification of treatment response with social interactions.
\newblock \textit{The Econometrics Journal}, \textbf{16} S1--S23.

\bibitem[{Ogburn et~al.(2020)Ogburn, Shpitser and Lee}]{Ogburn2020}
\textsc{Ogburn, E.~L.}, \textsc{Shpitser, I.} and \textsc{Lee, Y.} (2020).
\newblock Causal inference, social networks and chain graphs.
\newblock \textit{Journal of the Royal Statistical Society Series A: Statistics in Society}, \textbf{183} 1659--1676.

\bibitem[{Ogburn et~al.(2024)Ogburn, Sofrygin, Díaz and van~der Laan~and}]{Ogburn02012024}
\textsc{Ogburn, E.~L.}, \textsc{Sofrygin, O.}, \textsc{Díaz, I.} and \textsc{van~der Laan~and, M.~J.} (2024).
\newblock Causal inference for social network data.
\newblock \textit{Journal of the American Statistical Association}, \textbf{119} 597--611.

\bibitem[{O'Malley et~al.(2014)O'Malley, Elwert, Rosenquist, Zaslavsky and Christakis}]{omalley2014}
\textsc{O'Malley, A.~J.}, \textsc{Elwert, F.}, \textsc{Rosenquist, J.~N.}, \textsc{Zaslavsky, A.~M.} and \textsc{Christakis, N.~A.} (2014).
\newblock Estimating peer effects in longitudinal dyadic data using instrumental variables.
\newblock \textit{Biometrics}, \textbf{70} 506--515.

\bibitem[{Park and Kang(2022)}]{park2020efficient}
\textsc{Park, C.} and \textsc{Kang, H.} (2022).
\newblock Efficient semiparametric estimation of network treatment effects under partial interference.
\newblock \textit{Biometrika}, \textbf{109} 1015--1031.

\bibitem[{Perez-Heydrich et~al.(2014)Perez-Heydrich, Hudgens, Halloran, Clemens, Ali and Emch}]{perez2014assessing}
\textsc{Perez-Heydrich, C.}, \textsc{Hudgens, M.~G.}, \textsc{Halloran, M.~E.}, \textsc{Clemens, J.~D.}, \textsc{Ali, M.} and \textsc{Emch, M.~E.} (2014).
\newblock Assessing effects of cholera vaccination in the presence of interference.
\newblock \textit{Biometrics}, \textbf{70} 731--741.

\bibitem[{Puelz et~al.(2022)Puelz, Basse, Feller and Toulis}]{puelz2022graph}
\textsc{Puelz, D.}, \textsc{Basse, G.}, \textsc{Feller, A.} and \textsc{Toulis, P.} (2022).
\newblock A graph-theoretic approach to randomization tests of causal effects under general interference.
\newblock \textit{Journal of the Royal Statistical Society Series B: Statistical Methodology}, \textbf{84} 174--204.

\bibitem[{Qu et~al.(2024)Qu, Xiong, Liu and Imbens}]{qu2024}
\textsc{Qu, Z.}, \textsc{Xiong, R.}, \textsc{Liu, J.} and \textsc{Imbens, G.} (2024).
\newblock Semiparametric estimation of treatment effects in observational studies with heterogeneous partial interference.
\newblock \eprint{2107.12420}, \urlprefix\url{https://arxiv.org/abs/2107.12420}.

\bibitem[{Rosenbaum(2007)}]{Rosenbaum01032007}
\textsc{Rosenbaum, P.~R.} (2007).
\newblock Interference between units in randomized experiments.
\newblock \textit{Journal of the American Statistical Association}, \textbf{102} 191--200.

\bibitem[{Rubin(1972)}]{rubnin1972}
\textsc{Rubin, D.} (1972).
\newblock Estimating causal effects of treatments in experimental and observational studies.
\newblock \textit{ETS Research Bulletin Series}, \textbf{1972} i--31.

\bibitem[{Sherman and Shpitser(2018)}]{Sherman2018}
\textsc{Sherman, E.} and \textsc{Shpitser, I.} (2018).
\newblock Identification and estimation of causal effects from dependent data.
\newblock \textit{Advances in neural information processing systems}, \textbf{31}.

\bibitem[{Shpitser et~al.(2024)Shpitser, Park, Tchetgen and Andrews}]{shpitser2024}
\textsc{Shpitser, I.}, \textsc{Park, C.}, \textsc{Tchetgen, E.~T.} and \textsc{Andrews, R.} (2024).
\newblock Modeling interference via symmetric treatment decomposition.
\newblock \eprint{1709.01050}, \urlprefix\url{https://arxiv.org/abs/1709.01050}.

\bibitem[{Sobel(2006)}]{sobel2006randomized}
\textsc{Sobel, M.~E.} (2006).
\newblock What do randomized studies of housing mobility demonstrate? causal inference in the face of interference.
\newblock \textit{Journal of the American Statistical Association}, \textbf{101} 1398--1407.

\bibitem[{Sofrygin and van~der Laan(2017)}]{Sofrygin2017}
\textsc{Sofrygin, O.} and \textsc{van~der Laan, M.~J.} (2017).
\newblock Semi-parametric estimation and inference for the mean outcome of the single time-point intervention in a causally connected population.
\newblock \textit{Journal of Causal Inference}, \textbf{5} 20160003.

\bibitem[{Studeny(2013)}]{studeny2013bayesian}
\textsc{Studeny, M.} (2013).
\newblock Bayesian networks from the point of view of chain graphs.
\newblock \textit{arXiv preprint arXiv:1301.7414}.

\bibitem[{Sävje(2023)}]{savje2023}
\textsc{Sävje, F.} (2023).
\newblock Causal inference with misspecified exposure mappings: separating definitions and assumptions.
\newblock \textit{Biometrika}, \textbf{111} 1--15.

\bibitem[{Sävje et~al.(2021)Sävje, Aronow and Hudgens}]{savje2021}
\textsc{Sävje, F.}, \textsc{Aronow, P.} and \textsc{Hudgens, M.} (2021).
\newblock Average treatment effects in the presence of unknown interference.
\newblock \textit{Annals of Statistics}, \textbf{49} 673--701.

\bibitem[{Tchetgen~Tchetgen et~al.(2021)Tchetgen~Tchetgen, R.~Fulcher and Shpitser}]{auto-g}
\textsc{Tchetgen~Tchetgen, E.~J.}, \textsc{R.~Fulcher, I.} and \textsc{Shpitser, I.} (2021).
\newblock Auto-g-computation of causal effects on a network.
\newblock \textit{Journal of the American Statistical Association}, \textbf{116} 833--844.

\bibitem[{Tchetgen~Tchetgen and VanderWeele(2012)}]{TchetgenTchetgen2012}
\textsc{Tchetgen~Tchetgen, E.~J.} and \textsc{VanderWeele, T.~J.} (2012).
\newblock On causal inference in the presence of interference.
\newblock \textit{Statistical Methods in Medical Research}, \textbf{21} 55--75.

\bibitem[{VanderWeele et~al.(2012)VanderWeele, Vandenbroucke, Tchetgen~Tchetgen and Robins}]{vanderweele2012}
\textsc{VanderWeele, T.~J.}, \textsc{Vandenbroucke, J.~P.}, \textsc{Tchetgen~Tchetgen, E.~J.} and \textsc{Robins, J.~M.} (2012).
\newblock A mapping between interactions and interference: Implications for vaccine trials.
\newblock \textit{Epidemiology}, \textbf{23} 285--292.

\bibitem[{Verbitsky-Savitz and Raudenbush(2012)}]{VerbitskySavitz2012}
\textsc{Verbitsky-Savitz, N.} and \textsc{Raudenbush, S.~W.} (2012).
\newblock Causal inference under interference in spatial settings: A case study evaluating community policing program in chicago.
\newblock \textit{Epidemiologic Methods}, \textbf{1} 107--130.

\bibitem[{Weinstein and Nevo(2024)}]{weinstein2024}
\textsc{Weinstein, B.} and \textsc{Nevo, D.} (2024).
\newblock Causal inference with misspecified network interference structure.
\newblock \eprint{2302.11322}, \urlprefix\url{https://arxiv.org/abs/2302.11322}.

\bibitem[{Weinstein and Nevo(2025)}]{weinstein2025}
\textsc{Weinstein, B.} and \textsc{Nevo, D.} (2025).
\newblock Bayesian estimation of causal effects using proxies of a latent interference network.
\newblock \eprint{2505.08395}, \urlprefix\url{https://arxiv.org/abs/2505.08395}.

\end{thebibliography}

\end{document}